\documentclass[journal,twoside,web]{ieeecolor}
\usepackage{generic}
\usepackage{cite}
\usepackage{amsmath,amssymb,amsfonts}
\usepackage{algorithmic}
\usepackage{graphicx}
\usepackage{hyperref}
\hypersetup{hidelinks=true}
\usepackage{textcomp}
\usepackage{comment}
\def\BibTeX{{\rm B\kern-.05em{\sc i\kern-.025em b}\kern-.08em
    T\kern-.1667em\lower.7ex\hbox{E}\kern-.125emX}}

\usepackage[linesnumbered,ruled,vlined]{algorithm2e}
\usepackage{subcaption}
\newtheorem{theorem}{Theorem}

\usepackage{cleveref}
\let\labelindent\relax
\usepackage{enumitem}

\newcommand{\varIID}{\hat{\zeta}^2_{\mathrm{iid}}}

\newcommand{\varConsistent}{\hat{\zeta}^2}

\SetCommentSty{mycommfont}

\begin{document}
\title{Consistent Variance Estimation for Q-Function Estimators in Finite-Horizon MDP Tree Search}
\author{~~~Zhenyu Yue, Jie Xu, \IEEEmembership{Senior Member, IEEE}, Chun-Hung Chen, \IEEEmembership{Fellow, IEEE}, Hadi El-Amine, and Michael C. Fu, \IEEEmembership{Fellow, IEEE} 
\thanks{This work was supported in part by the U.S. National Science Foundation under Grants 2123683, 2228603, 2411248, 2436217. }
\thanks{
Zhenyu Yue and Michael C. Fu are with the University of Maryland, College Park, MD 20742 USA (e-mail: zyyue@umd.edu, mfu@umd.edu).}
\thanks{Jie Xu, Chun-Hung Chen, Hadi El-Amine are with George Mason University, Fairfax, VA 22030 USA (e-mail: jxu13@gmu.edu, cchen9@gmu.edu, helamine@gmu.edu).}
\thanks{This work has been submitted to the IEEE for possible publication. Copyright may be transferred without notice, after which this version may no longer be accessible.}
}

\maketitle

\begin{abstract}
We study the variance of Q-function estimators in finite-horizon, finite-state Markov decision process (MDP) tree search. We show that the variance decomposes into three components attributed to the immediate reward collected, probabilistic state transitions, and uncertainty in future state value function estimates. Using this decomposition, we show that the sample variance estimator based on the assumption of i.i.d. paths is biased, underestimating the true variance, and the bias does not vanish in the limit. We then propose a recursive variance estimator that is consistent. To enable efficient storage and computation, we derive an equivalent implementation of the recursive estimator using only node-local statistics that can be iteratively updated. This consistent variance estimator is integrated into two Monte Carlo Tree Search (MCTS) sampling procedures for finite-horizon MDPs. In numerical examples from inventory control and kidney paired donation matching, the new estimator improves the performance of the MCTS algorithm relative to a baseline that uses the i.i.d.-based sample variance estimator.
\end{abstract}


\section{Introduction}
\label{sec:introduction}
\IEEEPARstart{R}{einforcement} learning and simulation-based planning are widely used tools for sequential decision-making under uncertainty~\cite{sutton2018reinforcement,bertsekas1996neuro,powell2007approximate}. We focus on tree-structured algorithms for finite-horizon Markov decision processes (MDPs), which are generally more sample-efficient than policy-gradient algorithms, although at the expense of increased algorithmic complexity and higher computation cost. Here, by ``tree-structured,'' we mean methods---such as Monte Carlo Tree Search (MCTS)---that operate on a search tree and update value estimates via backpropagation based on the Bellman equation. In our setting, we are not restricted to full single-path rollouts: at each step, the algorithm may branch further or terminate a path early, and these decisions can be chosen adaptively. In these algorithms, and in reinforcement learning more broadly, uncertainty estimates play a central role in determining the quality of value estimates and guiding sample allocations.

Uncertainty estimates in tree-structured algorithms can be grouped into two major categories: range-based and variance-based methods. Range-based methods assume that rewards are bounded, typically with known bounds. Hoeffding's inequality provides the basis to construct confidence bounds for value estimators based on the sample size and reward bounds. These confidence bounds are the foundations of widely used algorithms such as UCB (for multi-armed bandit problems), UCT (UCB for trees), and variants such as UCBVI-CH~\cite{auer2002finite,KoSz06,azar2017minimax,browne2012survey,fu2024tutorial}.

An alternative to the bounded rewards assumption is finite variance of the random rewards. One variance-based method is optimal computing budget allocation adapted to MCTS (OCBA-MCTS), which uses the mean and variance of state-action value estimators (also known as Q-function estimators) to determine the sampling allocation that approximately maximizes the probability of correctly selecting the best action~\cite{li2019monte,li2022optimal}. OCBA-MCTS has also been shown to outperform UCT in simulation case studies on inventory control, tic-tac-toe, and Othello~\cite{fu2024tutorial,li2019monte,li2022optimal}. Other variance-aware methods combine range and variance information. Among them, UCB1-Tuned is a heuristic variant of UCB1 that incorporates empirical variance in its exploration term~\cite{auer2002finite}. A more sophisticated UCB variant, UCBVI-BF, employs Bernstein–Freedman concentration inequalities to incorporate variance-dependent exploration bonuses for value estimates and improved regret bounds~\cite{azar2017minimax}.


Despite the benefits of using variance in MDP algorithms, to the best of our knowledge, existing algorithms do not use a variance estimator that propagates uncertainty period by period for Q-function estimates in finite-horizon tree-structured algorithms. OCBA-MCTS uses the sample variance of path returns, treating them as if they were independent and identically distributed (i.i.d.), which ignores correlations induced by shared value function estimates across paths. In Section~\ref{sec:variance}, we show that this approach underestimates the actual variance in Q-function estimators. Existing variance analysis methods largely focus on the distributional properties of the cumulative return under a fixed policy~\cite{sobel1982variance,sobel1994meanvariance,mannor_tsitsiklis_2011} to inform risk-averse decision making, e.g., penalizing high-variance policies, but do not provide variance estimates needed to guide sample allocation within tree search.

In this paper, we analyze the variance of Q-function estimators in finite-horizon MDPs and develop a consistent estimator together with an efficient backpropagation method for tree-structured algorithms. Section~\ref{sec:variance} presents the main theorem that decomposes the Q-function estimator variance into three components arising from uncertainty in rewards, transition dynamics, and next-period state value function estimates. We propose a new recursive variance estimator in Section~\ref{sec:consistent} and show its consistency. To efficiently implement the recursive variance estimator, Section~\ref{sec:efficient} presents an iterative node-local updating scheme that allows for efficient storage and updating. Section~\ref{sec:LFOCBA} shows how the new consistent estimator can be incorporated into OCBA-MCTS and UCB1-Tuned. In Section~\ref{sec:examples}, we report numerical results that demonstrate the improved performance of both OCBA-MCTS and UCB1-Tuned for an inventory control problem and a kidney paired donation problem. We conclude the paper and discuss possible future research in Section~\ref{sec:conclusion}.

\section{Variance Propagation in MDPs}
\label{sec:variance}
Consider an undiscounted finite-horizon, finite-state MDP with components
\((S_{\mathrm{phys}}, \{\mathcal{A}_t\}, \{P_t\}, \{R_t\}, H)\). The process starts from an initial state \(X_0=x_0\in S_{\mathrm{phys}}\) and is in a state \(X_t\in S_{\mathrm{phys}}\) at each period \(t=0,1,\dots,H\), with the horizon being \(H\in\mathbb{N}\). While in a state $X_t$, an action \(a_t\) is chosen from a finite feasible action set
\(\mathcal{A}_t(X_t)\). Conditional on \(X_t = x\) and \(a_t = a\), the next-period state
\(X_{t+1}\) is a random element of \(S_{\mathrm{phys}}\) with transition
probabilities
\(P_t(x'\mid x,a)=\mathbb{P}(X_{t+1}=x'\mid X_t=x,A_t=a)\), and the one-step reward
is a random variable \(R_t(x,a,X_{t+1})\). Since we do not assume time homogeneity, the action sets, transition probabilities, and rewards may depend on time. A deterministic Markov policy is a sequence
\(\pi=(\pi_0,\dots,\pi_{H-1})\) of functions on \(S_{\mathrm{phys}}\), where
each \(\pi_t\) satisfies \(\pi_t(x)\in \mathcal{A}_t(x)\) for all \(x\in S_{\mathrm{phys}}\).
Define the optimal value functions
\[
V_t(x)
  := \max_{\pi} \mathbb{E}^{\pi}\!\left[ \sum_{u=t}^{H-1}
        R_u\bigl(X_u,\pi_u(X_u),X_{u+1}\bigr)
     \,\Big|\, X_t = x \right],
\]
where the expectation is taken over the random states and rewards generated by actions $a_u = \pi_u(X_u) $, and
\[
\begin{aligned}
Q_t(x,a)
  &:= \mathbb{E}\!\bigl[
        R_t(x,a,X_{t+1}) + V_{t+1}(X_{t+1}) 
        | X_t = x, A_t = a \bigr],
\end{aligned}
\] with terminal condition \(V_H(x)=0\) for all \(x\in S_{\mathrm{phys}}\).
Our objective is to choose a policy that attains the maximum expected cumulative reward
from \(x_0\); we call any policy achieving this an optimal policy,
that is, any \(\pi^{*} \in \operatorname*{arg\,max}_{\pi} \mathbb{E}^{\pi}\!\left[ \sum_{u=0}^{H-1}
        R_u\bigl(X_u,\pi_u(X_u),X_{u+1}\bigr)
     \,\Big|\, X_0 = x_0 \right]\).

For notational convenience we now absorb time into the state. Let
\(S := \{(t,x) : t=0,\dots,H,\ x\in S_{\mathrm{phys}}\}\) denote the augmented
state space, and let \(\mathcal{T} := \{H\}\times S_{\mathrm{phys}}\) be the set of
terminal states. For \(s=(t,x)\in S\) we write \(\mathcal{A}(s):=\mathcal{A}_t(x)\). For
\(s=(t,x)\notin\mathcal{T}\) and \(a\in \mathcal{A}(s)\), let
\(S'=(t+1,X_{t+1})\) denote the random next augmented state corresponding to the
physical next-period state \(X_{t+1}\), with
\(\mathbb{P}(S'=(t+1,x')\mid s,a)=P_t(x'\mid x,a)\), and define
\(R(s,a,S'):=R_t(x,a,X_{t+1})\).
We then write \(V(s):=V_t(x)\) and \(Q(s,a):=Q_t(x,a)\). Bellman’s principle of optimality implies that these optimal value functions satisfy, for every nonterminal state $s$,
\[ Q(s,a) = \mathbb{E}\bigl[R(s,a,S') + V(S')\bigr],\] \[ V(s) = \max_{a\in \mathcal{A}(s)} Q(s,a), \] with \(V(s)=0\) for all \(s\in\mathcal{T}\). Thus, any policy that, at every nonterminal state $s$, chooses an action $a\in A(s)$ attaining $\max_{a\in A(s)}Q(s,a)$ is optimal. From this point on, we work with the
augmented notation and suppress explicit time indices.

When the expectations in the Bellman equations cannot be evaluated analytically, a
common numerical approach is to approximate them by Monte Carlo sampling and apply
dynamic programming backward from terminal states, using \(V(s)=0\) on
\(\mathcal{T}\). For any state \(s\) and action \(a\in \mathcal{A}(s)\), let
\(\{s'_i\}_{i=1}^{k}\) enumerate the possible successor states with
\(P(s'_i\mid s,a)>0\). Suppose there are \(N(s,a)\) Monte Carlo samples
\(\bigl(r^{(j)}, s'_{I_j}\bigr)\), \(j=1,\dots,N(s,a)\), where
\(I_j\in\{1,\dots,k\}\) is drawn according to \(\mathbb{P}(I_j=i)=P(s'_i\mid s,a)\)
and \(r^{(j)}\) is an independent draw with conditional distribution equal to the distribution
of \(R(s,a,s'_{I_j})\). Applying the Bellman equations leads to the recursive
estimators
\begin{subequations}\label{eq:Q-est}
\begin{align}
  \hat{Q}(s,a)&=\frac{1}{N(s,a)}
      \sum_{j=1}^{N(s,a)}
      \bigl(r^{(j)}+\hat{V}(s'_{I_j})\bigr), \label{eq:Q-est-a}\\
  \hat{V}(s)  &=\max_{a\in \mathcal{A}(s)}\hat{Q}(s,a), \label{eq:Q-est-b}
\end{align}
\end{subequations}
with \(\hat{V}(s)=0\) for \(s\in\mathcal{T}\). Equation~\eqref{eq:Q-est-a} has the
form of a sample average. It is straightforward to show \cite{li2022optimal} that
under MCTS selection policies such as UCT and OCBA, \(\hat{Q}(s,a)\) is a consistent
estimator of the true value \(Q(s,a)\). To apply a selection policy such as
OCBA, the variance of \(\hat{Q}(s,a)\) must be estimated. However, \(\hat{Q}(s,a)\)
is not the sample average of \(N(s,a)\) i.i.d.\ terms because of the dependence
introduced by \(\hat{V}(s'_{I_j})\), which is itself an estimator and therefore a
random variable. In the following, we derive a recursive equation to
estimate the variance of \(\hat{Q}(s,a)\).

\subsection{Previous Work}
The sample average form of $\hat{Q}(s,a)$ seems to suggest that its variance,  which we denote by
\begin{equation} \label{eq:zeta}
\zeta^{2}(s,a)=\operatorname{Var} \bigl[\hat{Q}(s,a)\bigr], 
\end{equation}
can also be estimated using the sample variance formula: 
\begin{equation}\label{eq:Var-approx}
  \varIID(s,a)=
  \frac{\displaystyle
        \sum_{j=1}^{N(s,a)}
        \bigl[r^{(j)}+\hat{V}(s'_{I_j})
          -\hat{Q}(s,a)
        \bigr]^{2}}
       {N(s,a)\bigl(N(s,a)-1\bigr)} .
\end{equation}
%

The MCTS-OCBA algorithm introduced in \cite{li2019monte,li2022optimal} took this approach. However, $\varIID(s,a)$ is biased whenever different samples transition to the same successor state ($s_{I_{j}}'=s_{I_{l}}'$ with $j\neq l$), because the terms $r^{(j)}+\hat{V}(s_{I_{j}}')$ and $r^{(l)}+\hat{V}(s_{I_{l}}')$ are positively correlated due to the shared value estimate $\hat{V}$, leading to an underestimation of the variance, and this bias does not vanish in the limit. Although rewards and state transitions are sampled independently each time, the value function estimate $\hat{V}(s'_{i})$ is shared across all samples that transition into state $s'_i$. Consider an extreme case where the state  transition is fixed, i.e., $s'_{I_j} = s'_i$ for all $j$. We then have $\hat{Q}(s,a) = \hat{V}(s'_i)+\sum_{j=1}^{N(s,a)} {r^{(j)}}/{N(s,a)}$. Therefore, $\varIID(s,a)$ is only determined by the variance of the rewards $\{r^{(j)}\}$ and the variance in the estimate of $\hat{V}(s')$ is omitted. Thus, (\ref{eq:Var-approx}) is an underestimate of $\zeta^2(s,a)$. Because this underestimation persists even as $N(s,a) \to \infty$ (unless the next-period value variance vanishes), the estimator defined in \eqref{eq:Var-approx} is neither unbiased nor consistent.

\subsection{Q-function Estimation Variance}
To calculate the theoretical variance of $\hat{Q}(s,a)$, we consider $N(s,a)$ samples $(R^{(j)},s'_{I_j}), j = 1,..., N(s,a)$, and introduce the following notation: 

\begin{itemize}
  \item $p_i  =  \mathbb{P}\bigl(S' = s'_i \mid s,a\bigr)$, the probability of
        transitioning to state $s'_i$ when action $a$ is taken in state $s$,
        for $i=1,\dots,k$.
  \item $\alpha_i$ and $\tau_i^{2}$ — mean and variance of the conditional
        reward $R(s,a,s'_i)$.

\end{itemize}

Throughout the variance analysis, assume that the sample sizes used to
construct the successor value estimates $\{\hat V(s'_i)\}_{i=1}^k$ are chosen
without reference to the current-period samples
$\{(I_j,R^{(j)})\}_{j=1}^{N(s,a)}$ (for example, by fixing a predetermined
number of rollouts at each successor). This condition is sufficient for the
independence assumptions in Theorem~\ref{thm:mdp_variance}. In many practical
online MCTS implementations, however, the samples used to update $\hat V$ and
$\hat Q$ are intertwined, so the variance expression in
Theorem~\ref{thm:mdp_variance} should be viewed as that of an idealized
two-stage sampling scheme.

\begin{theorem}[Variance of $\hat{Q}(s,a)$]  \label{thm:mdp_variance}
If the successor states value estimators $\{\hat{V}(s'_i)\}_{i=1}^k$ are independent across $i$, and jointly independent of the current period transitions and rewards $\{I_j,\,R^{(j)}\}_{j=1}^{N(s,a)}$, then 
\begin{IEEEeqnarray}{rCl}
\zeta^{2}(s,a)
&=& \sum_{i=1}^{k} p_i^{2}\,\operatorname{Var}(\hat{V}(s'_i)) \nonumber \\
&& + \frac{1}{N(s,a)} \Bigl[\sum_{i=1}^{k} p_i(1-p_i)\operatorname{Var}(\hat{V}(s'_i)) + \sum_{i=1}^{k} p_i\,\tau_i^{2} \nonumber \\
&&  + \sum_{i=1}^{k} p_i \bigl(\alpha_i + \mathbb{E}[\hat{V}(s'_i)] - \mathbb{E}[\hat{Q}(s,a)]\bigr)^{2} \Bigr]. \label{eq:Var-truth}
\end{IEEEeqnarray}
\end{theorem}
\vspace{3pt}
\begin{proof}
We introduce some shorthand notation, $Y_i := \hat{V}(s'_i), R^{(j)}_{I_j} := R^{(j)}(s,a,s'_{I_j}), m := N(s,a), I := \{I_j, j = 1,...,m\}$. So we have $\hat{Q}(s,a)=\sum_{j=1}^{m}R^{(j)}_{I_j}/m + \sum_{j=1}^{m} Y_{I_j}/m$. Further denote $N_i :=  \#\{j: I_j=i\}, S_R:=\sum_{j=1}^{m}R^{(j)}_{I_j}, S_Y:=\sum_{j=1}^{m}Y_{I_j}$. Because of the assumed independence of the reward and the value function estimates, 
\begin{equation*} 
    \zeta^2(s,a) = \frac{\operatorname{Var}(S_R)}{m^2} +  \frac{\operatorname{Var}(S_Y)}{m^2} + \frac{2\operatorname{Cov}(S_R, S_Y)}{m^2}. 
\end{equation*}

We first calculate the variance of $S_R$. Using the conditional variance formula, we have $ \operatorname{Var}(S_R)=\mathbb{E} \bigl[\operatorname{Var}(S_R \mid I)\bigr]+\operatorname{Var} \bigl(\mathbb{E}[S_R \mid I]\bigr)$. For the first term, because $\operatorname{Var}(S_R \mid I) = \sum_{j=1}^{m} \tau_{I_j}^{2}$, we have $\mathbb{E} \bigl[\operatorname{Var}(S_R \mid I)\bigr]
=\sum_{i=1}^{k}m p_i\tau_i^{2}$. To calculate the second term, we first notice that $\mathbb{E}[S_R \mid I]=\sum_{i=1}^{k}\alpha_i N_i$. Because $N_i$ is multinomial with probability $p_i$, applying the covariance matrix of a multinomial random vector yields 
 \[
  \operatorname{Var} \bigl(\mathbb{E}[S_R \mid I]\bigr)
  =m \left[\sum_{i=1}^{k}p_i\alpha_i^{2}
            - \left(\sum_{i=1}^{k}p_i\alpha_i\right)^{2}\right].
\]
Hence, we have
\begin{equation} \label{eq:varsr}
  \operatorname{Var}(S_R)
  =m \left[\sum_{i=1}^{k}p_i\tau_i^{2}
            +\sum_{i=1}^{k}p_i\alpha_i^{2}
            - \left(\sum_{i=1}^{k}p_i\alpha_i\right)^{2}\right].
\end{equation}
We next derive the variance of $S_Y$, again using the conditional variance formula, $\operatorname{Var}(S_Y) =\mathbb{E} \bigl[\operatorname{Var}(S_Y \mid Y)\bigr] + \operatorname{Var} \bigl(\mathbb{E}[S_Y \mid Y]\bigr)$. It is easy to see that $\operatorname{Var} \bigl(\mathbb{E}[S_Y \mid Y]\bigr) = m^{2}\sum_{i=1}^{k}p_i^{2}\operatorname{Var}(\hat{V}(s'_i))$. Applying the multinomial covariance matrix gives 
\[
  \operatorname{Var}(S_Y \mid Y)
  =m \left[\sum_{i=1}^{k}p_iY_i^{2}
            - \left(\sum_{i=1}^{k}p_iY_i\right)^{2}\right].
\]
Because of the independence assumption of $Y_i$, it is straightforward to show $\mathbb{E} \Bigl[\sum_{i=1}^{k}p_iY_i^{2}\Bigr] = \sum_{i=1}^{k}p_i(\operatorname{Var}(\hat{V}(s'_i))+\mathbb{E}[\hat{V}(s'_i)]^{2})$, $\mathbb{E} \Bigl[\left(\sum_{i=1}^{k}p_iY_i\right)^{2}\Bigr]
  =\sum_{i=1}^{k}p_i^{2}\operatorname{Var}(\hat{V}(s'_i)) + \left( \sum_{i=1}^{k} p_i \mathbb{E}[\hat{V}(s'_i)]\right)^{2}$. Combining all three terms, we have 
\begin{multline} \label{eq:varsy}
\operatorname{Var}(S_Y)
 = m\Bigl[\sum_{i=1}^{k} p_i(1-p_i)\operatorname{Var}(\hat{V}(s'_i)) + \sum_{i=1}^{k} p_i\mathbb{E}[\hat{V}(s'_i)]^{2}  \\
  - \Bigl(\sum_{i=1}^{k} p_i\mathbb{E}[\hat{V}(s'_i)]\Bigr)^{2}\Bigr] + m^{2}\sum_{i=1}^{k} p_i^{2}\operatorname{Var}(\hat{V}(s'_i)).
\end{multline}    

Finally, we compute $\operatorname{Cov}(S_R,S_Y)$ using the conditional covariance formula, $ \operatorname{Cov}(S_R,S_Y)  = \mathbb{E} \bigl[\operatorname{Cov}(S_R,S_Y \mid I)\bigr] + \operatorname{Cov} \bigl(\mathbb{E}[S_R \mid I], \mathbb{E}[S_Y \mid I]\bigr)$. Given \(I\), $S_R$ and $S_Y$ are independent, and thus the first term is zero. It is easy to observe that the second term is
\begin{multline} \label{eq:cov}
  \operatorname{Cov} \Bigl(\sum_{j=1}^{m}\alpha_{I_j}, \sum_{j=1}^{m}\mu_{I_j}\Bigr)= \\ 
  m \left[\sum_{i=1}^{k}p_i\alpha_i\mathbb{E}[\hat{V}(s'_i)] - \left(\sum_{i=1}^{k}p_i\alpha_i\right)     \left(\sum_{i=1}^{k}p_i\mathbb{E}[\hat{V}(s'_i)]\right)\right].
\end{multline}

Adding up (\ref{eq:varsr})-(\ref{eq:cov}) and rearranging terms, we obtain
\begin{multline} \label{eq:mvarq}
    m^2 \zeta^2(s,a) =  m^{2}\sum_{i=1}^{k} p_i^{2}\operatorname{Var}(\hat{V}(s'_i)) + \\ 
      m\left[\sum_{i=1}^{k} p_i\tau_i^{2} + \sum_{i=1}^{k} p_i(1-p_i)\operatorname{Var}(\hat{V}(s'_i)) \right] + \\
    m\left[ \sum_{i=1}^{k} p_i(\alpha_i+\mathbb{E}[\hat{V}(s'_i)])^{2} - \Bigl(\sum_{i=1}^{k} p_i(\alpha_i+\mathbb{E}[\hat{V}(s'_i)])\Bigr)^{2}\right].
\end{multline}
Because $\mathbb{E}[\hat{Q}(s,a)] = \sum_{i=1}^{k} p_i(\alpha_i+\mathbb{E}[\hat{V}(s'_i)])$, the last term of the right-hand side of (\ref{eq:mvarq}) can be rearranged as $\Bigl[\sum_{i=1}^{k} p_i\bigl(\alpha_i+\mathbb{E}[\hat{V}(s'_i)]-\mathbb{E}[\hat{Q}(s,a)]\bigr)^{2}\Bigr]$. Dividing both sides of (\ref{eq:mvarq}) by \(m^{2}\) yields the final result.
\end{proof}

We now make some observations on the composition of $\zeta^2(s,a)$, which reveals how different sources of randomness affect the quality of $\hat{Q}(s,a)$ as an estimator of $Q(s,a)$. 
\begin{itemize}
  \item Uncertainty in next period value estimates: 
        the term $\sum_{i=1}^{k} p_i^{2}\operatorname{Var}(\hat{V}(s'_i))$ propagates the uncertainty already present in the value estimates $\hat{V}(s'_i)$. Increasing $N(s,a)$ does not necessarily reduce this term, because not all samples would transition to state $s'_i$. 
  \item Variability in empirical probability estimates: 
        the term 
        $\sum_{i=1}^{k} p_i(1-p_i)\operatorname{Var}(\hat{V}(s'_i))/N(s,a)$  accounts for the variance in estimating the true transition probability $p_i$. This term is negligible, because it is the product of two vanishing quantities: the inverse sample size $1/N(s,a)$ and the successor variance $\operatorname{Var}(\hat{V}(s'_i))$.
  \item Random reward variance: the term $\sum_{i=1}^{k} p_i\tau_i^{2}/N(s,a)$ is the expected conditional variance of the immediate reward.
  \item State-transition variance: the term \(\sum_{i=1}^k p_i\bigl(\alpha_i+\mathbb{E}[\hat V(s'_i)]-\mathbb{E}[\hat Q(s,a)]\bigr)^2/{N(s,a)}\) is the variance from probabilistic state transitions. 
\end{itemize}

\section{Consistent Estimation} \label{sec:consistent}
We first derive the asymptotic bias of the approximate estimator $\varIID(s,a)$ and then propose a consistent recursive alternative.
\subsection{Asymptotic behavior of \texorpdfstring{$\varIID(s,a)$}{zeta\_approx}}
\label{subsec:approx-asymp}

\begin{theorem}\label{thm:2}
As \(N(s,a)\to\infty\),
\[
  \zeta^2(s,a)
     \;\longrightarrow\;
     \sum_{i=1}^{k} p_i^{2}\operatorname{Var}(\hat{V}(s'_i)),
  \qquad
  \varIID(s,a)
     \;\longrightarrow\; 0.
\]
\end{theorem}

\begin{proof}
The first limit follows directly from Theorem~\ref{thm:mdp_variance}. As $N(s,a) \to \infty$, the terms in the variance decomposition scaled by $1/N(s,a)$ vanish, leaving only the term $\sum_{i=1}^{k} p_i^{2}\operatorname{Var}(\hat{V}(s'_i))$.

For the second limit, consider a fixed realization of the successor estimators $\{\hat V(s'_i)\}$. Conditional on these values, the augmented samples $X_j = R^{(j)} + \hat V(s'_{I_j})$ are independent across $j$. In this case, $\varIID(s,a)$ is equal to the usual unbiased sample variance of the $X_j$ divided by $N(s,a)$, so its expectation is of order $N(s,a)^{-1}$ and hence $\varIID(s,a) \to 0$ as $N(s,a)\to\infty$.
\end{proof}

Consequently, with all other parameters fixed, the approximate variance estimator is asymptotically biased except for the last period of the MDP, with the bias given by
\begin{equation} \label{eq:vbias}
  \lim_{N(s,a)\to\infty} \operatorname{Var} \bigl[\hat{Q}(s,a)\bigr] - \varIID(s,a) = \sum_{i=1}^{k} p_i^{2}\operatorname{Var}(\hat{V}(s'_i)).
\end{equation}

\begin{theorem}\label{thm:approx-bias}
Fix a realization of the successor value estimates, denoted by $v = \{v_i\}_{i=1}^k$, where $\hat V(s'_i) = v_i$. Consider $N(s,a)$ samples collected at the current node $(s,a)$. We have
\begin{align}\label{eq:thm3-conditional-mean}
&N(s,a)\mathbb{E}\bigl[\varIID(s,a)
  \,\big|\, \{\hat V(s'_i)\} = \{v_i\}\bigr]\notag\\
&\quad=
\sum_{i=1}^k p_i \tau_i^2 
 +
\sum_{i=1}^k p_i \bigl(\alpha_i + v_i
  - \mathbb{E}[\hat Q(s,a)\mid v]\bigr)^2,
\end{align}
where we write $v = (v_1,\dots,v_k)$ and
\[
\mathbb{E}[\hat Q(s,a)\mid v]
= \sum_{i=1}^k p_i (\alpha_i + v_i).
\]
Moreover,
\begin{equation}\label{eq:thm3-convergence}
N(s,a)\,\varIID(s,a)
\xrightarrow{p}
N(s,a)\mathbb{E}\bigl[\varIID(s,a)
  \,\big|\, \{\hat V(s'_i)\} = \{v_i\}\bigr],
\end{equation}
as \(N(s,a)\to \infty\).
\end{theorem}

\begin{proof}
We apply the result of Theorem~\ref{thm:mdp_variance} by substituting $v_i$ for $\mathbb{E}[\hat{V}(s'_i)]$ and setting $\operatorname{Var}(\hat{V}(s'_i))=0$ as we condition on $v = \{v_i\}$. This leads to
\[
\zeta^2(s,a)\mid v
=
\frac{1}{N(s,a)}\left[
  \sum_{i=1}^k p_i \tau_i^2
  +
  \sum_{i=1}^k p_i \bigl(\alpha_i + v_i - \mathbb{E}[\hat Q(s,a)\mid v]\bigr)^2
\right].
\]

The estimate of the $Q$ function conditioning on $\hat{V}$ is $$\mathbb{E}[\hat{Q}(s,a)|v] =  \mathbb{E}[\hat{Q}(s,a)] + \sum_{i=1}^{k} p_i\bigl(v_i-\mathbb{E}[\hat{V}(s'_i)]\bigr).$$ 
Because $X_j=R^{(j)}+v_{I_j}$ are i.i.d. after conditioning on $\hat{V} = v$, we have
\begin{IEEEeqnarray}{rCl}
  N(s,a)\mathbb{E}[ \varIID(s,a) \mid v] &=& \operatorname{Var}[\hat{Q}(s,a) \mid v]  \nonumber \\
  &=&  N(s,a)\zeta^2(s,a)\mid v. \nonumber    
\end{IEEEeqnarray}

This yields the equality 
\begin{multline*}
N(s,a)\mathbb{E} \bigl[\varIID(s,a)
 \,\big|\,v\bigr]
= \sum_{i=1}^{k} p_i\tau_i^{2} \\
 + \sum_{i=1}^{k} p_i\bigl(\alpha_i+v_i-\mathbb{E}[\hat{Q}(s,a)|v]\bigr)^{2}.
\end{multline*}

Also due to the fact that $X_j$ are i.i.d. conditional on $v$, the sample variance estimator is consistent: $N(s,a) \varIID(s,a)\xrightarrow{p} N(s,a)\mathbb{E}\bigl[\varIID(s,a)\,\big|\,\hat{V}(s'_i) = v_i\bigr]$ as $N(s,a) \to \infty$.
\end{proof}

Comparing Theorem~\ref{thm:approx-bias} with the variance decomposition
in Theorem~\ref{thm:mdp_variance}, we see that, conditional on $v$,
$N(s,a)\,\varIID(s,a)$ asymptotically recovers only
the reward-variance term $\sum_{i=1}^k p_i \tau_i^2$ and the
state-transition variance term
$\sum_{i=1}^k p_i \bigl(\alpha_i + v_i - \mathbb{E}[\hat Q(s,a)\mid v]\bigr)^2$.
The two terms involving uncertainty in the successor value estimates,
$N(s,a)\sum_{i=1}^k p_i^2 \operatorname{Var}(\hat V(s'_i))$ and
$\sum_{i=1}^k p_i(1-p_i)\operatorname{Var}(\hat V(s'_i))$, are entirely
omitted. Thus, whenever there is non-negligible uncertainty in the
next-period value estimates, $\varIID(s,a)$
systematically underestimates $\zeta^2(s,a)$, which motivates the
recursive correction introduced in the next subsection.

\subsection{A consistent estimator \texorpdfstring{$\varConsistent(s,a)$}{zeta\_cst}}

Comparing Theorem~\ref{thm:mdp_variance} and Theorem~\ref{thm:approx-bias}, we observe that the approximate estimator captures two of the four terms in $\zeta^2(s,a)$, but leaves out the following two terms due to the variance in the value function estimates: \(   \sum_{i} p_i^{2}\operatorname{Var}(\hat{V}(s'_i)) + N(s,a)^{-1}\sum_{i} p_i(1-p_i)\operatorname{Var}(\hat{V}(s'_i))\), where the second term vanishes for large \(N(s,a)\), and the first term is equal to the variance contributed by the next-period value function estimates. We therefore define a new recursive variance estimator:
\begin{subequations}\label{eq:var-cst}
\begin{align}
\varConsistent(s,a)
  &= \varIID(s,a)
     + \sum_{i=1}^{k} \hat{p}_i^{2}\,\varConsistent(s'_i),\\
 \varConsistent(s,a)
  &= 0 \quad \text{for}\  s\in\mathcal{T},\\
\varConsistent(s'_i)
  &= \varConsistent\bigl(s'_i,\,a^{*}\bigr),\\
a^{*}
  &= \arg\max_{a'} \hat{Q}(s'_i,a'),\\
\hat{p}_i
  &= \frac{N(s,a,s'_i)}{N(s,a)},
\end{align}
\end{subequations}
where \(N(s,a,s'_i)\) counts how many samples from \((s,a)\) transitioned to state \(s'_i\). In particular, we do not assume \( N(s,a) = \sum_{i=1}^{k} N(s'_{i}) \). This accounts for state sharing, where successor states \(s'_i\) may accumulate samples from other parts of the search tree (or independent simulations) beyond just the transitions observed from \((s,a)\).
\begin{theorem}
\label{thm:cst-bias}
Suppose that, for each successor state \(s'_i\), the true mean, variance, and transition probability are used, i.e.,
\(
   \hat{V}(s'_i)=\mathbb{E}[\hat{V}(s'_i)],\;
   \varConsistent(s'_i)=\operatorname{Var}(\hat{V}(s'_i)),\;
   \hat{p}_i=p_i.
\)
Then
\[N(s,a) \mathbb{E} \bigl[\bigl(\varConsistent(s,a) -\zeta^2 (s,a)\bigr)\bigr] = -\sum_{i=1}^{k}p_i(1-p_i)\operatorname{Var}(\hat{V}(s'_i)).\] 

  
In particular, when the state is terminal, $s\in \mathcal{T}$, the right-hand-side is $0$.
\end{theorem}
\begin{proof}
Use Theorem ~\ref{thm:approx-bias} to calculate the expected value of $\varIID$ with the perfect estimates, and substitute into part (a) and (b) of ~\eqref{eq:var-cst}. 
Subtract the variance provided in Theorem~\ref{thm:mdp_variance} and multiply the difference by $N(s,a)$. Three of the four terms from the theoretical variance will cancel each other out, leaving the remaining one term.
\end{proof}


As a result of Theorem~\ref{thm:cst-bias}, the bias of
$\varConsistent(s,a)$ is of order $1/N(s,a)$, and its
magnitude shrinks further as $\operatorname{Var}(\hat V(s'_i))$
decreases. In particular, if
$\operatorname{Var}(\hat V(s'_i)) \to 0$, then
\[
\mathbb{E}\bigl[\varConsistent(s,a)
  - \zeta^2(s,a)\bigr] \to 0.
\]

\begin{theorem}
Suppose that as $N(s,a)\rightarrow\infty$, the following assumptions hold:
\begin{enumerate}[label=\roman*)]
    \item  the value estimator $\hat{V}(s_{i}')$ converges in probability to its mean $\mathbb{E}[\hat{V}(s_{i}')]$,
    \item  the successor states variance vanishes, i.e., $Var(\hat{V}(s_{i}')) \rightarrow 0$, and
    \item  the successor states variance estimates converge at a rate satisfying $N(s,a)\sum_{i=1}^{k}|\hat{p}_{i}^{2}\varConsistent(s_{i}')-p_{i}^{2}Var(\hat{V}(s_{i}'))| \xrightarrow{p} 0$.
\end{enumerate}
Then the variance estimator $\varConsistent(s,a)$ is consistent, satisfying:
\begin{equation}
N(s,a)(\varConsistent(s,a)-\zeta^{2}(s,a)) \xrightarrow{p} 0.
\end{equation}
\end{theorem}

\begin{proof}
The convergence for the terminal state given in (\ref{eq:var-cst}b) is a trivial result of the convergence of \(N\varIID\), so we focus on the two terms in (\ref{eq:var-cst}a). Let \(m=N(s,a)\). Using (\ref{eq:var-cst}a) together with Theorem~\ref{thm:mdp_variance}, write \(m\bigl(\varConsistent(s,a)-\zeta^2(s,a)\bigr)=A+B-C\), where
\(A :=
m \varIID(s, a)
  - \sum_{i=1}^k p_i \tau_i^2
  - \sum_{i=1}^k p_i \bigl(
      \alpha_i + \mathbb{E}[\hat V(s'_i)]
      - \mathbb{E}[\hat Q(s, a)]
    \bigr)^2\);
\(B :=
m \sum_{i=1}^k \bigl(
  \hat p_i^{\,2}\,\varConsistent(s'_i)
  - p_i^{\,2}\,\operatorname{Var}(\hat V(s'_i))
\bigr)\);
\(
C :=
\sum_{i=1}^k p_i (1-p_i)\,\operatorname{Var}(\hat V(s'_i)).\)

For \(A\): conditioning on the fixed vector \(\{\hat V(s'_i)\}\) makes the augmented samples \(X_j=R^{(j)}+\hat V(s'_{I_j})\) independent across \(j\), so \(m\,\varIID(s,a)\) converges (in probability) to \(\sum_{i=1}^{k} p_i\tau_i^{2} + \sum_{i=1}^{k} p_i\bigl(\alpha_i+\mathbb{E}[\hat V(s'_i)]-\mathbb{E}[\hat Q(s,a)]\bigr)^2\). With \(\hat V(s'_i)\xrightarrow{p}\mathbb{E}[\hat V(s'_i)]\) by assumption~(i), Slutsky's theorem yields \(A\xrightarrow{p}0\).

For \(B\): by assumption~(iii), we have \(m\sum_{i=1}^{k}\bigl|\hat p_i^{2}\,\varConsistent(s'_i)-p_i^{2}\,\operatorname{Var}(\hat V(s'_i))\bigr|\xrightarrow{p}0\), hence \(B\xrightarrow{p}0\).

For \(C\): since \(\operatorname{Var}(\hat V(s'_i)) \to 0\) by assumption~(ii), the deterministic term \(C\) converges to 0.

Combining these results, we conclude \(m\bigl(\varConsistent(s,a)-\zeta^2(s,a)\bigr)\xrightarrow{p} 0\).
\end{proof}

The above theorem establishes the consistency of the variance estimator $\varConsistent(s,a)$. However, one should interpret assumption~(iii) carefully. It implies that for the recursive estimator to be consistent at the current node, the variance estimates at the successor nodes must converge sufficiently fast. Specifically, the cumulative error from estimating successor variances and transition probabilities (scaled by the current sample size $N(s,a)$) must vanish. If the tree is too deep or the successor estimates are too noisy, this error could accumulate and bias the root variance estimate.

\section{Efficient Estimation via Node-Local Statistics}
\label{sec:efficient}

The variance estimator derived in Section~\ref{sec:consistent} is consistent, but using it inside iterative algorithms like MCTS raises a practical issue: simulations are generated \emph{sequentially} and a na\"ive implementation would recompute every state's mean and variance after each new sample, which is computationally infeasible for large and deep trees. In MCTS terminology, we lack an efficient ``backpropagation'' rule to update the consistent variance estimates when a new simulation sample is collected.

Before detailing the specific algorithms, we clarify the objective of the backpropagation scheme. While standard tree search requires only the value estimate $\hat{Q}(s,a)$, the variance-based algorithms considered here fundamentally require a second statistic: the local variance component, denoted as $s^{2}(s,a):=N(s,a)\varIID(s,a)$. Because the consistent variance estimator $\varConsistent$ (defined in \eqref{eq:var-cst}) is a function of $\varIID$, maintaining $s^{2}$ (and therefore $\varIID$) is necessary. The goal of the following updating schemes, therefore, is to efficiently update this pair---$\hat{Q}(s,a)$ and $s^{2}(s,a)$---whenever a new sample is collected.

In the following subsections, we examine three updating schemes. We begin with the baseline \emph{Exact Updating} which stores full sample histories to ensure correctness but suffers from high computational cost. We contrast this with \emph{Summary Statistics}, an efficient approach that yields incorrect estimates by relying on outdated successor values (lagging updates). Finally, we demonstrate that our proposed \emph{Node-Local Statistics} achieves the best of both worlds: it is computationally efficient and mathematically equivalent to the Exact Updating method.

\subsection{Three Updating Schemes} \label{subsec:updating-schemes}

For a state-action $(s,a)$ with sample size $N(s,a)$, by substituting $\varIID(s,a) = s^2(s,a)/N(s,a)$ into Equation \eqref{eq:Var-approx}, we see that $s^2$ is the sample variance of terms $x_j = r^{(j)} + \hat{V}(s'_{I_j}), j=1,...,N(s,a)$. We call $x_j$ augmented samples.

\subsubsection{Method 1: Exact Updating} \label{subsec:exact-update}
This method ensures correctness by storing the full history of rewards $\{r^{(j)}\}$ for each action node $(s,a)$. When evaluating the node's value and variance, the algorithm iterates through this stored history and adds the \emph{current} value estimate of the successor state to each reward to form the current vector of augmented samples $\{x_j = r^{(j)} + \hat{V}(s'_{I_j})\}$. The mean and sample variance of this vector are then computed to obtain $\hat{Q}$ and $s^2$.

While exact, this method requires a bottom-up recomputation whenever a sample is added. Because the value estimate of a successor state changes over time, the augmented samples must be reconstructed from scratch to ensure correctness. Consequently, the computational cost at each node grows linearly with the number of simulations ($O(N)$), which becomes prohibitively expensive for deep trees with large sample counts.

\subsubsection{Method 2: Iteration with Summary Statistics} \label{subsec:approx-update}
To avoid storing full histories, this scheme updates statistics iteratively using running estimates. When a new sample is added, the algorithm sums the immediate reward $r$ with the \emph{current} successor value estimate to create the augmented sample \(x_{N+1} = r +\hat{V}(s')\). It then updates the running mean $\tilde{Q}$ and variance proxy $\tilde{s}^2$ as follows:
\begin{align*}
\tilde{Q}_{N+1} &= \frac{N\tilde{Q}_N+x_{N+1}}{N+1},\\
\tilde{s}^2_{N+1} &= \frac{N-1}{N}\tilde{s}^2_{N}+\frac{(x_{N+1}-\tilde{Q}_N)^{2}}{N+1}.
\end{align*}
The fundamental flaw of this method is \emph{update lag}. The summary statistics permanently absorb the value $\hat{V}(s')$ as it existed at the moment of the visit. Future updates to $\hat{V}(s')$ are not retroactively applied to the summary statistics $\tilde{Q}$ or $\tilde{s}^2$. As a result, these estimators diverge from the true sample mean $\hat{Q}$ and variance $s^2$. This discrepancy introduces inaccuracies, as we demonstrate in the example in Section \ref{subsec:example-comparison}.
\subsubsection{Method 3: Proposed Node-Local Statistics} \label{subsec:suff-update}
To address the computational cost of the exact updating method and the inaccuracy of the summary statistics method, we propose a new node-local statistics updating method. The main idea is to store reward statistics separately from the value estimates and combine them only when computing $\hat Q(s,a)$ and $s^2(s,a)$. For an action node $(s,a)$ with $\hat{Q}, s^2$, and $N$ statistics, each possible successor state $s'_i, 1\leq i\leq k$ has three statistics (reward sample mean $\hat{\alpha}_i$, reward sample variance $\hat{\tau}_i^{2}$, and transition frequency $N_i$) stored and requires iterative updates when a new sample \((s,a,r,s'_i)\) arrives:
\begin{subequations}\label{eq:tau-alpha}
\begin{align}
\hat{\alpha}_i &\leftarrow
\frac{N_i\,\hat{\alpha}_i + r}{N_i+1},\\
\hat{\tau}_i^{2} &\leftarrow
\frac{N_i-1}{N_i}\hat{\tau}_i^{2}
      +\frac{(r-\hat{\alpha}_i)^{2}}{N_i+1},\\
N_i &\leftarrow N_i+1,\quad
N\leftarrow N+1.
\end{align}
\end{subequations}
These estimators are kept local and update only with samples transitioning from $(s,a)$ to the specific next state $s'_i$, instead of using summary statistics that depend on the time-varying $\hat{V}$ of successor states. Reward updates affect only \(\hat{\alpha}_i,\hat{\tau}_i^{2}\), while value propagation is performed separately, as detailed in the following theorem.

\begin{theorem}Consider a state-action pair \((s,a)\) with total sample size \(N\).
For each successor state \(s'_i\), let \(N_i, \hat{\alpha}_i, \text{and } \hat{\tau}_i^2\) be the statistics defined in \eqref{eq:tau-alpha}, and let \(\hat{V}(s'_i)\) denote the current state value estimate. Then, 
\begin{align}
\hat{Q}(s,a) &= \sum_{i=1}^{k}\frac{N_i}{N}\bigl(\hat{\alpha}_i+\hat{V}(s'_i)\bigr),
\label{eq:val-update-Q}\\
\varIID(s,a)
 &= \frac{1}{N(N-1)}
  \Bigl[
\sum_{i=1}^{k}(N_i-1)\hat{\tau}_i^{2}\notag\\ 
&+\sum_{i=1}^{k}N_i
\bigl(\hat{\alpha}_i+\hat{V}(s'_i)-\hat{Q}(s,a)\bigr)^{2}
  \Bigr], \label{eq:var-approx-s}\\
s^2(s,a) &= N\,\varIID(s,a). \label{eq:s2-final}
\end{align} 
\end{theorem}

\begin{proof}
The derivation of the mean estimator $\hat{Q}(s,a)$ in \eqref{eq:val-update-Q} is trivial from the definition of the sample mean partitioned by successor states.
For the variance estimator, recall that the baseline $\varIID(s,a)$ is defined as the sample variance of the $N$ augmented samples $x_j = r^{(j)} + \hat{V}(s'_{I_j})$:
\[
    \varIID(s,a) = \frac{1}{N(N-1)} \sum_{j=1}^N \bigl(x_j - \hat{Q}(s,a)\bigr)^2. 
\]
We decompose the summation by grouping samples according to their successor state index $i \in \{1, \dots, k\}$.
Let $G_i = \{j: I_j = i\}$ be the set of indices $j$ such that the transition was to $s'_i$.
For a fixed successor $i$, let $\bar{x}_i = \hat{\alpha}_i + \hat{V}(s'_i)$ be the local sample mean.
We add and subtract $\bar{x}_i$ inside the square and expand the term for group $i$:
\begin{align*}
    \sum_{j \in G_i} &(x_j - \hat{Q}(s,a))^2 \\
    &= \sum_{j \in G_i} \bigl[ (x_j - \bar{x}_i) + (\bar{x}_i - \hat{Q}(s,a)) \bigr]^2 \\
    &= \underbrace{\sum_{j \in G_i} (x_j - \bar{x}_i)^2}_{A} + \underbrace{\sum_{j \in G_i} (\bar{x}_i - \hat{Q}(s,a))^2}_{B} \\
    &\quad + \underbrace{2(\bar{x}_i - \hat{Q}(s,a)) \sum_{j \in G_i} (x_j - \bar{x}_i)}_{C}.
\end{align*}
We analyze the three terms individually:
\begin{itemize}
    \item $A$:  This is the sum of squared deviations from the local mean, which relates to the stored local variance $\hat{\tau}_i^2$:
    \[ A = \sum_{j \in G_i} (r^{(j)} - \hat{\alpha}_i)^2 = (N_i - 1)\hat{\tau}_i^2.
\]
    \item $B$: Since the term inside the summation is constant with respect to $j$, we simply multiply by the count $N_i$:
    \[ B = N_i \bigl(\hat{\alpha}_i + \hat{V}(s'_i) - \hat{Q}(s,a)\bigr)^2.
\]
    \item $C$: $\sum_{j \in G_i} x_j = N_i \bar{x}_i$ and $\sum_{j \in G_i} \bar{x}_i = N_i \bar{x}_i$, so $C = 0$.
\end{itemize}
Summing $A + B$ over all groups $i=1,\dots,k$ and dividing by $N(N-1)$ yields \eqref{eq:var-approx-s}. Equation \eqref{eq:s2-final} follows immediately from the definition of $s^2$.
\end{proof}

In other words, \eqref{eq:val-update-Q}--\eqref{eq:s2-final} give the same $\hat Q(s,a)$ and $s^2(s,a)$ as if we had stored and processed all augmented samples $X_j = R^{(j)} + \hat V(s'_{I_j})$, while requiring only node-local statistics.
The highest tree-updating computational cost for updating (\ref{eq:val-update-Q}) and (\ref{eq:var-approx-s}) is \(O(k)\), where $k$ is the number of successor states (dependent on $s$ and $a$) and is independent of sample path count.

\subsection{Comparative Analysis via Illustrative Example}
\label{subsec:example-comparison}

To compare these methods, we use a numerical example of an MDP with a horizon of two periods.
In this example, the state transition is deterministic while the rewards are random.
We specifically analyze the update performed when a ninth simulated sample path is added to a tree built from eight prior paths.
The simulation history is listed below. Each path is represented as a tuple $(a, r, s', a', r')$:
\begin{enumerate}
    \item $(a_1, -1, s_1, a_{11}, -1)$
    \item $(a_1, 2, s_1, a_{11}, 3)$
    \item $(a_1, 2, s_1, a_{12}, 4)$
    \item $(a_1, 1, s_1, a_{12}, 0)$
    \item $(a_2, 1, s_2, a_{21}, 0)$
    \item $(a_2, 1, s_2, a_{21}, 2)$
    \item $(a_2, 1, s_2, a_{22}, -2)$
    \item $(a_2, 1, s_2, a_{22}, 0)$
  
  \item \textbf{New Path: $(a_1, 3, s_1, a_{11}, 5.5)$}
\end{enumerate}

\begin{figure}
  \centering
  \includegraphics[width=\linewidth]{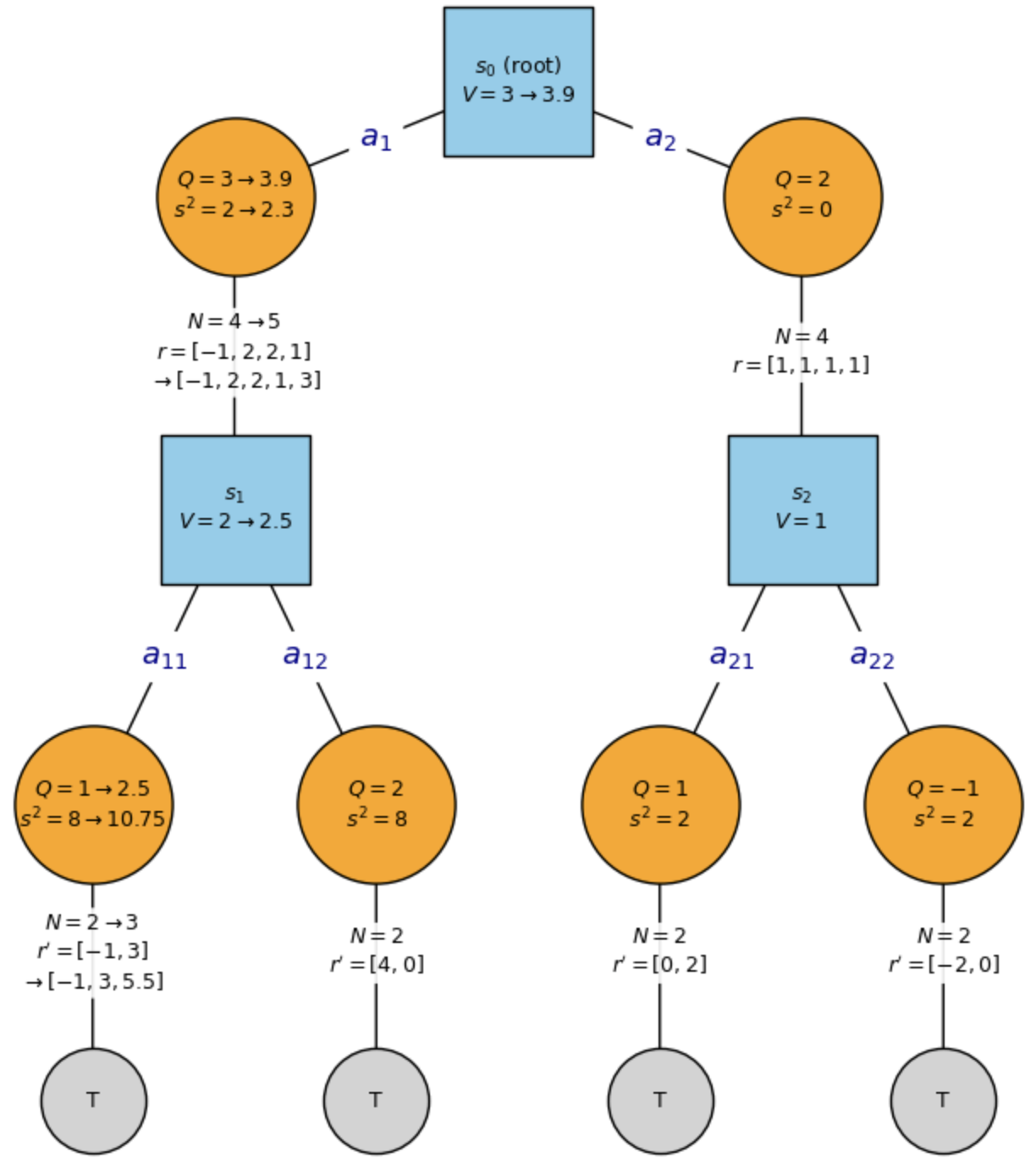}
  \caption{Exact Updating ($N_{\text{root}}=9$): Tree state updated (\textit{old} $\to$ \textit{new}) with the ninth sample path.}
  \label{fig:Exact_Tree}
\end{figure}

Fig.~\ref{fig:Exact_Tree} shows the tree state after updating with the ninth sample path under the \textbf{Exact Updating} scheme.
This serves as the ground truth for our comparison. The correct updated value for $\hat{Q}(s_0,a_1)$ is 3.9.
\begin{figure}
  \centering
  \includegraphics[width=\linewidth]{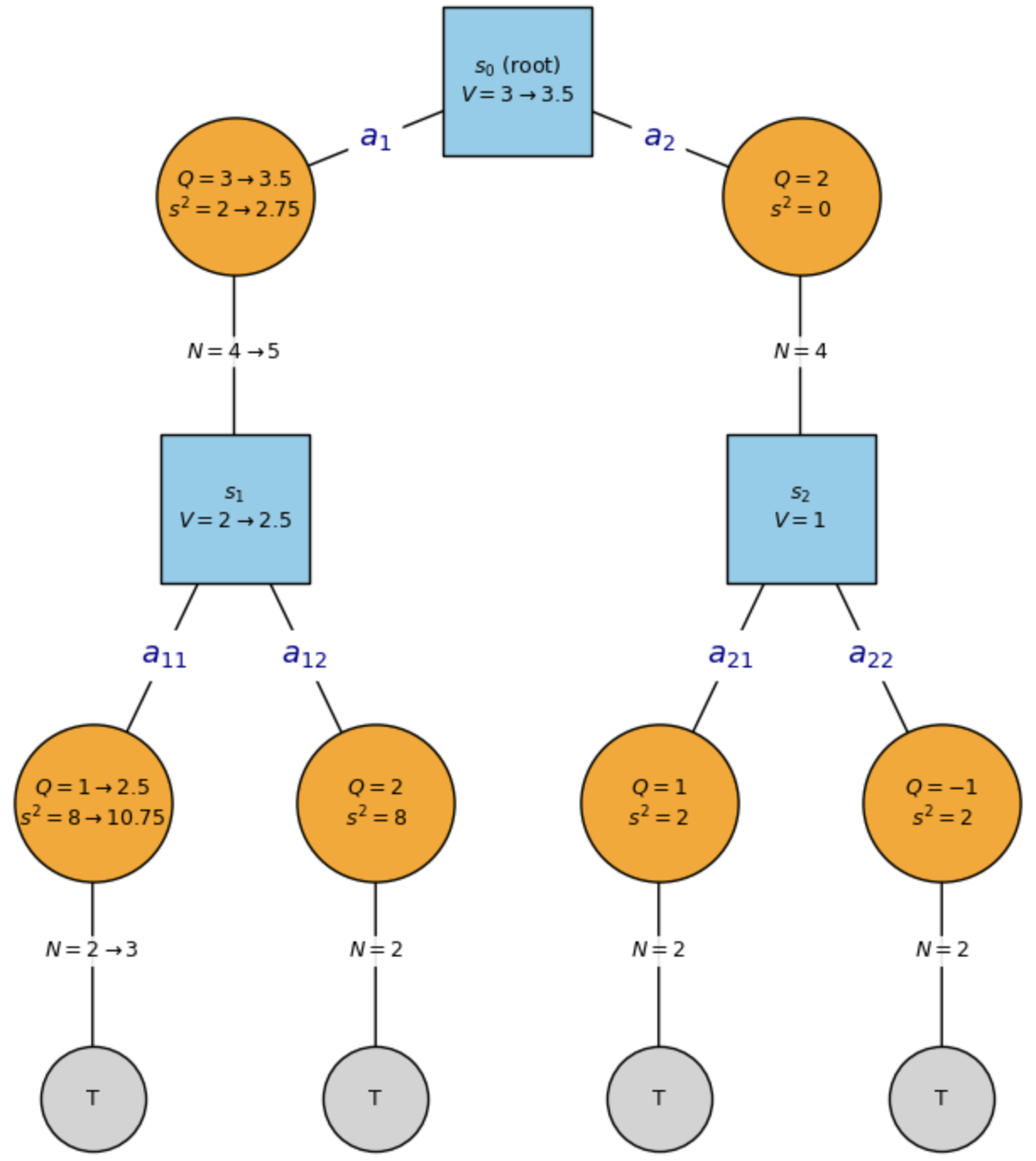}
  \caption{Summary Statistics ($N_{\text{root}}=9$): Tree state updated (\textit{old} $\to$ \textit{new}) with the ninth sample path converted to augmented rewards \((s_0, a_1, x=3+2.5, s_1, a_{11}, x=5.5)\).}
  \label{fig:Approximate_Tree}
\end{figure}

Fig.~\ref{fig:Approximate_Tree} depicts the update using \textbf{Summary Statistics}.
We intentionally initialized the tree with the correct values from the first eight paths to demonstrate that even given a perfect initial state, this rule introduces bias in a single step.
For action $a_1$ with immediate reward $r = 3$, the augmented sample is $x = r + V = 3 + 2.5 = 5.5$.
The value of $a_1$ is updated by $Q \leftarrow (4 \times 3 + 5.5)/(4 + 1) = 3.5$.
Because the previously stored summary statistics at node $a_1$ were computed using the obsolete estimate of $\hat{V}(s_1)$, this updating scheme produces a biased $\hat{Q}(s_0,a_1) = 3.5$ instead of the correct value of 3.9 as in the Exact Updating scheme shown in Fig.~\ref{fig:Exact_Tree}.
\begin{figure}
  \includegraphics[width=\linewidth]{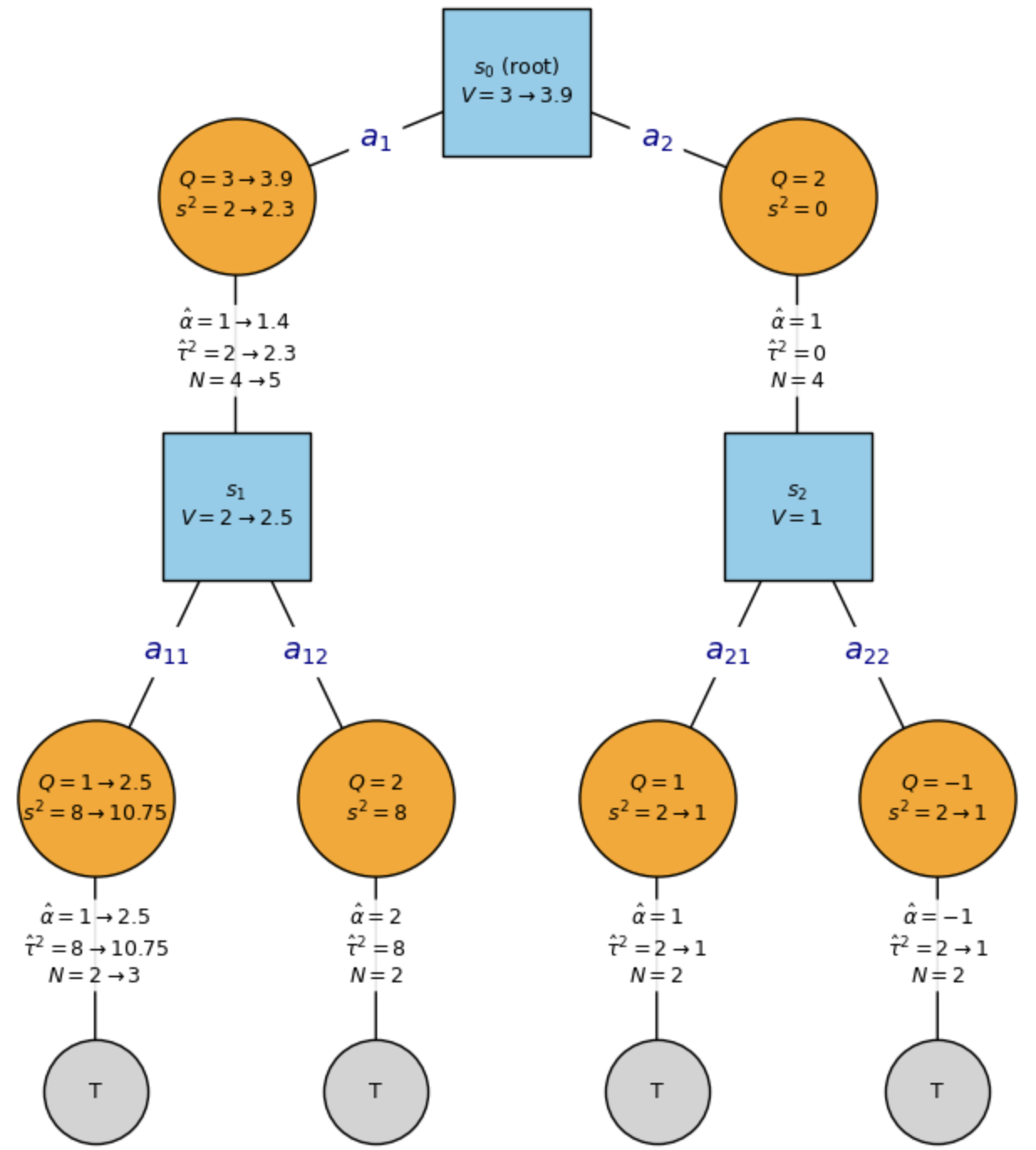}
\caption{Node-Local Statistics ($N_{\text{root}}=9$): Tree state updated (\textit{old} $\to$ \textit{new}) with the ninth sample path \((s_0, a_1, r=3, s_1, a_{11}, r=5.5)\).}
  \label{fig:Sufficient_Tree}
\end{figure}

Fig.~\ref{fig:Sufficient_Tree} shows the update using the proposed \textbf{Node-Local Statistics}.
The average reward of $a_{11}$ is updated by $\hat{\alpha} \leftarrow (N\hat{\alpha}+r)/(N+1) = (2+5.5)/3 = 2.5$, giving $Q = \bar{r} = 2.5$, which then updates $\hat{V}(s_1) = 2.5$.
The average reward of $a_{1}$ is updated by $\hat{\alpha} \leftarrow (N\hat{\alpha}+r)/(N+1) = (4\times 1+3)/5 = 1.4$.
The key difference between node-local statistics and Summary Statistics updating is that here $\hat{Q}(a_1)$ is updated using the sum of the reward average and the updated $\hat{V}(s_1)$, $\hat{Q}(s_0,a_1) = \bar{r} + \hat{V}(s_1) = 1.4 + 2.5 = 3.9$.
The value of $s_0$ is thus updated to 3.9. This result agrees with the exact method.

\subsection{Node-Local Statistics with Random Transitions}
\label{subsec:random-transition-example}

We now consider the case where the MDP may transition into multiple successor states.
We apply the updating rules defined in Equations \eqref{eq:val-update-Q}--\eqref{eq:s2-final} to the example shown in Fig \ref{fig:Method_Comparison}.
\begin{figure}
\includegraphics[width=0.8 \linewidth]{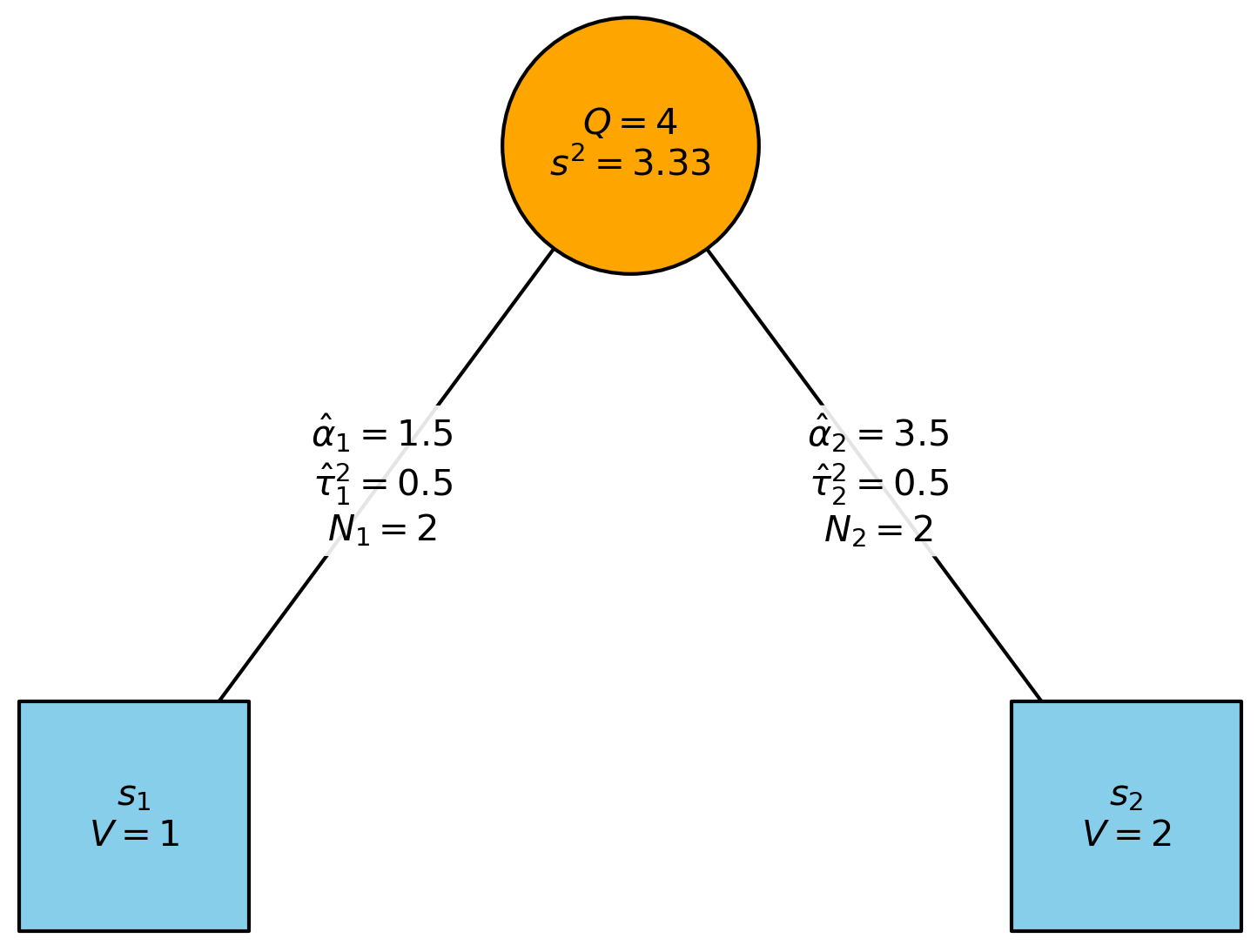}
\vspace{-0.4em}
\caption{Node-Local Statistics with random transitions into two successor states. Samples: \((r=1,s_1), (2,s_1), (3,s_2), (4,s_2)\).}
\label{fig:Method_Comparison}
\end{figure}

For transitions to \(s_1\), the rewards are \(r_1=[1,2]\), giving
\((\hat{\alpha}_1=1.5,\;\hat{\tau}_1^{2}=0.5)\);
for \(s_2\), the rewards are \(r_2=[3,4]\), giving
\((\hat{\alpha}_2=3.5,\;\hat{\tau}_2^{2}=0.5)\).  
With two visits to each state, \(N_1=N_2=2\) and
\(N=4\).
We thus apply \eqref{eq:val-update-Q}--\eqref{eq:s2-final} to get $\hat{Q}=0.5(1.5+1)+0.5(3.5+2)=4$, and $s^2 = \frac{1}{3}\hat{\tau}_1^{2}
  + \frac{2}{3}\bigl(\hat{\alpha}_1+\hat{V}(s_1)-\hat{Q}\bigr)^{2}+\frac{1}{3}\hat{\tau}_2^{2}
  +\frac{2}{3}\bigl(\hat{\alpha}_2+\hat{V}(s_2)-\hat{Q}\bigr)^{2} = 10/3$.

We now verify that this result agrees with the result from the Exact Updating method.
The four sample paths in Fig \ref{fig:Method_Comparison} yield augmented samples $x = r+\hat{V}$, resulting in values $\{1+1 = 2, 1+ 2 = 3, 3+2 = 5, 4+2 = 6\}$.
The sample mean is 4 and the sample variance is $10/3 \approx 3.33$.
In this example, the Exact Updating and the Node-Local Statistics variance are identical, demonstrating that we can fully recover the correct variance information using this method.
In summary, our proposed method is mathematically equivalent to the exact method by correctly accounting for updates in value estimates, while being nearly as efficient as the summary statistics method.
\section{Integration with MCTS Algorithms}\label{sec:LFOCBA}
In this section, we detail the integration of the consistent variance estimator and the node-local statistics method from Section~\ref{sec:efficient} into the MCTS framework. MCTS is a tree-structured algorithm widely used for finite-horizon MDPs~\cite{KoSz06,Co06,Fu18,Fu19} and serves as the core search mechanism in systems such as AlphaGo and AlphaZero~\cite{Si17}.

\begin{figure}[htbp]
\centering
  \includegraphics[width=0.5\textwidth]{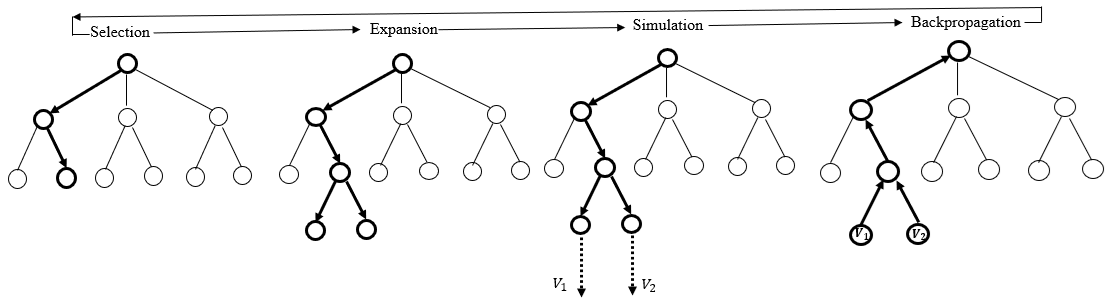}
  \caption{One MCTS iteration: selection $\rightarrow$ expansion $\rightarrow$ simulation $\rightarrow$ backpropagation.}
  \label{fig:mcts}
\end{figure}

A standard MCTS iteration consists of four phases: selection, expansion, simulation, and backpropagation (see Fig.~\ref{fig:mcts}).
During the \emph{selection} phase, the algorithm traverses the tree from the root by applying a selection policy until it reaches a leaf node that has at least one untried action.
In the \emph{expansion} phase, a new child node corresponding to one of these untried actions is added.
The \emph{simulation} (also known as roll-out) phase then executes a default policy from this new node to a terminal state to collect a sequence of rewards. The sum of these rewards constitutes a realization of the cumulative reward.
Finally, the \emph{backpropagation} phase updates the statistics along the traversed path using this cumulative reward.
This process repeats until the computational budget is exhausted, at which point the action with the highest estimated $Q$-value at the root is returned.

The selection policy balances exploration and exploitation to determine sample allocation. Variance-sensitive selection rules rely on accurate estimates of the $Q$-function's uncertainty to guide this allocation. We integrate the consistent variance estimator into two such rules: OCBA and UCB1-Tuned.

\subsection{The OCBA Selection Rule}
OCBA-MCTS~\cite{li2022optimal} allocates the sampling budget to maximize the probability of correct selection (PCS). At a node $s$, let $a^* = \arg\max_{a} Q(s,a)$ be the current best action and $\Delta_a = Q(s,a^*) - Q(s,a)$ be the value gap for suboptimal actions $a \neq a^*$. The asymptotically optimal sample allocation satisfies the following conditions:
\begin{align}
\frac{N(s,a^*)}{\zeta^2(s,a^*)} &= \sum_{a\neq a^*}\frac{N(s,a)}{\zeta^2(s,a)}, \label{eq:ocba-mdp-bestrest}\\
\frac{\zeta^2(s,a)}{\Delta_a^{2}} &= \frac{\zeta^2(s,a')}{\Delta_{a'}^{2}},\quad \forall a,a'\neq a^*. \label{eq:ocba-mdp-rest}
\end{align}
In the proposed implementation, the theoretical quantities $a^*$, $\Delta_a$, and $\zeta^2$ are replaced by their empirical counterparts $\hat{a}^*$, $\hat{\Delta}_a$, and the consistent variance estimator $\varConsistent(s,a)$ (or $\varIID(s,a)$ for the baseline).

To satisfy these conditions heuristically, the algorithm evaluates \eqref{eq:ocba-mdp-bestrest} at each step. If the ratio for the best action $N(s,\hat{a}^*)/\hat{\zeta}^2(s,\hat{a}^*)$ is less than the sum of the ratios for the suboptimal actions, the next sample is allocated to $\hat{a}^*$. Otherwise, the sample is allocated to the suboptimal action $u \neq \hat{a}^*$ that maximizes the ratio $\hat{\zeta}^2(s,u)/\hat{\Delta}_u^2$, thereby driving the allocation towards the equality condition in \eqref{eq:ocba-mdp-rest}.

\subsection{The UCB1-Tuned Selection Rule}
The UCB1-Tuned rule~\cite{auer2002finite} minimizes cumulative regret by incorporating the variance of the reward distribution into the Upper Confidence Bound (UCB) exploration term. Adapted for MCTS, the rule selects the action that maximizes:
\begin{align}
I(s,a) &= Q(s,a) \nonumber \\
&+ \sqrt{\frac{\ln N_p(s)}{N(s,a)} \min\Bigl(\tfrac{1}{4}, \sigma^2(s,a) + \sqrt{\tfrac{2\ln N_p(s)}{N(s,a)}}\Bigr)},
\end{align}
where $N_p(s)$ is the total visit count for state $s$, and $\sigma^2(s,a)$ represents the variance of the cumulative rewards observed after taking action $a$. We approximate $\sigma^2(s,a)$ using $N(s,a)\hat{\zeta}^2(s,a)$, utilizing either the consistent or baseline variance estimator.

Since UCB1-Tuned assumes rewards are bounded within $[0,1]$, we apply dynamic normalization.
We track the minimum and maximum $Q$-value estimates, $\hat{Q}_{\min}(t)$ and $\hat{Q}_{\max}(t)$, observed at each decision period $t$ (corresponding to the time step $t$ in the finite horizon) across the entire tree.
The range is approximated by $\hat W(t)=\max(\hat{Q}_{\max}(t)-\hat{Q}_{\min}(t), \varepsilon)$. The algorithm selects the action maximizing the normalized score:
\begin{align} \label{eq:ucb1}
\hat I(s,a) &= \frac{\hat Q(s,a)-\hat{Q}_{\min}(t)}{\hat W(t)} \nonumber \\
& \hspace{-2em} +\sqrt{\frac{\ln N_p(s)}{N(s,a)} \min\Bigl(\tfrac{1}{4}, \frac{N(s,a)\hat{\zeta}^2(s,a)}{\hat W^2(t)} + \sqrt{\tfrac{2\ln N_p(s)}{N(s,a)}}\Bigr)}.
\end{align}
Calculating $\hat W(t)$ globally for each period reduces the volatility associated with local, sparse estimates at individual nodes.
Furthermore, because the potential cumulative reward naturally decreases as the period $t$ approaches the horizon $H$, a single global normalization constant would excessively dampen exploration at later periods.
Periodwise normalization ensures the exploration term remains sensitive to the specific value scale at each step of the horizon.

\subsection{Implementation and Pseudocode}
Algorithm~\ref{alg:mcts_c} contains the pseudocode for the proposed consistent MCTS variants (denoted as ``\textsc{-C}''), utilizing the node-local statistics derived in Section~\ref{sec:efficient} to update the consistent variance estimator $\varConsistent$. The implementation stores the following statistics, all initialized to 0:
\begin{itemize}
    \item \textbf{State}: visit count $N(s)$ and value estimate $\hat{V}(s)$.
    \item \textbf{Action}: visit count $N(s,a)$, $Q$-value estimate $\hat{Q}(s,a)$, baseline variance $\varIID(s,a)$, and consistent variance $\varConsistent(s,a)$.
    \item \textbf{Transition}: transition count $N(s,a,s')$, reward sample mean $\hat{\alpha}(s,a,s')$, and reward sample variance $\hat{\tau}^2(s,a,s')$.
\end{itemize}
We set the warm-up parameter $n_0(s)$ to 4 for the root node $s_0$ and 2 for all other nodes. This ensures sufficient samples are collected before applying the variance-based selection rules.

\begin{algorithm}[htbp]
\caption{MCTS with Node-Local Statistics}
\label{alg:mcts_c}
\DontPrintSemicolon

\KwIn{budget $N$, root state $s_0$, selection rule $\texttt{rule}\in\{\text{``OCBA''},\ \text{``UCB1-Tuned''}\}$}
\KwOut{updated node, action, and transition statistics}
\For{$n \leftarrow 1$ \KwTo $N$}{
  $s \leftarrow s_0,\ \textit{path} \leftarrow [\,]$\;
  \tcp{Forward pass: selection / expansion / simulation}
  \While{$s$ is nonterminal}{
    \If{$s$ is not in tree}{
      Initialize statistics for $s$ ($N=0, \hat{V}=0$)\;
      \ForEach{$a \in A(s)$}{
        Initialize statistics for $(s,a)$ ($N=0, \hat{Q}=0, \varIID=0, \varConsistent=0$)\;
      }
    }
    $A_{\text{warm}} \leftarrow \{\,a \in A(s) : N(s,a) < n_0(s)\,\}$\;
    \eIf{$A_{\text{warm}} \neq \varnothing$}{
      Choose $a \in A_{\text{warm}}$ uniformly at random\;
    }{
      Choose $a$ via $\texttt{rule}$ using $\hat Q(s,\cdot)$ \& $\varConsistent(s,\cdot)$ \tcp*{Eqs. \eqref{eq:ocba-mdp-bestrest}--\eqref{eq:ocba-mdp-rest} or \eqref{eq:ucb1}}
    }
    $(s', r) \leftarrow \texttt{sample}(s,a)$ \tcp*{Simulate 1 step}
    Append $(s,a,r,s')$ to $\textit{path}$, $s \leftarrow s'$\;
  }
  \tcp{Backward pass: node-local backpropagation}
  \ForEach{$(s,a,r,s')\ \in\ \texttt{reversed}(\textit{path})$}{
    \If{$(s,a,s')$ is not in tree}{
       Initialize transition statistics ($N=0, \hat{\alpha}=0, \hat{\tau}^2=0$)\;
    }
    Update transition stats $N, \hat{\alpha}, \hat{\tau}^{2}$ via \eqref{eq:tau-alpha}\;
    $N(s,a) \leftarrow N(s,a) + 1$\;
    Update $\hat Q(s,a)$ via \eqref{eq:val-update-Q} using all successors\;
    Update $\varIID(s,a)$ via \eqref{eq:var-approx-s} using all successors\;
    Update $\varConsistent(s,a)$ via \eqref{eq:var-cst} using $\varIID(s,a)$ and $\varConsistent(s')$\;
    $N(s) \leftarrow N(s) + 1$, $\hat V(s) \leftarrow \max_{u} \hat Q(s,u)$\;
    $a^\star \leftarrow \arg\max_{u} \hat Q(s,u)$, $\varConsistent(s) \leftarrow \varConsistent(s,a^\star)$\;
  }
}
\end{algorithm}
\section{Numerical Examples} \label{sec:examples}
We compare two procedures for updating MCTS under the selection rules
OCBA and UCB1-Tuned: a baseline procedure, which follows the
OCBA-MCTS implementation of~\cite{li2022optimal} and
uses the summary-statistics method of Section~\ref{subsec:approx-update}
with the approximate variance estimator
$\varIID(s,a)$ in \eqref{eq:Var-approx}, and the
proposed consistent procedure (in Algorithm 1), which uses the node-local statistics
method of Section~\ref{subsec:suff-update} together with the consistent
variance estimator $\varConsistent(s,a)$ defined in
\Cref{eq:var-cst}. We conduct experiments on four test problems. The first problem is the inventory control problem used to test OCBA-MCTS in \cite{li2022optimal}. The other two problems change the parameters for the inventory cost structure and horizon to increase the difficulty of the problem. The last test problem is based on the kidney-paired donation (KPD) problem in \cite{yue2024kidney}. The performance of the procedures is evaluated by comparing their solutions with the true optimal solutions obtained by DP.

We evaluate performance of different MCTS algorithms via the empirical PCS (of the first action taken) as a function of cumulative sample size (the number of simulated sample paths). The empirical PCS is the percentage of runs in which the algorithm selects the optimal action for the initial state, computed over 1{,}000 independent replications (3{,}000 for the KPD problem). 

We also report the mean absolute percentage error (MAPE) of the variance estimators at the root node \(s_0\).
Specifically, we track the accuracy of the local variance component, measured by the scaled quantity
\(
  N(s_0, a^*)\,\varIID(s_0, a^*).
\)
We focus on this local component for two key reasons.
First, the total recursive variance \(\varConsistent(s,a)\) depends on the sample counts of all successor nodes throughout the tree.
These counts vary dynamically during the search, creating a ``moving target'' that makes convergence analysis noisy and difficult to interpret.
In contrast, Theorem~\ref{thm:approx-bias} provides a stable, sample-size-independent limit for the local component conditioned on the current value estimates.
Second, and more critically, this metric isolates the efficacy of the \emph{Node-Local Statistics} updating scheme proposed in Section~\ref{sec:efficient}.
The baseline \emph{Summary Statistics} method (Section~\ref{subsec:approx-update}) introduces bias, because it cannot retrospectively update the variance contribution of previously collected samples when successor value estimates \(\hat{V}(s')\) change.
By validating that our method accurately recovers the local variance, we demonstrate that the proposed updating scheme successfully eliminates this source of error, providing a necessary foundation for the consistent recursive estimator.

Results are plotted in Figs \ref{fig:IC_prob}–\ref{fig:KPD_val}. In the legend labels, the suffix “\textsc{-C}” denotes the algorithm using the consistent value/variance estimator.

\subsection{Inventory Control}

\subsubsection{Inventory-1} Consider a finite-horizon inventory model with i.i.d. demand, holding costs, and lost-sale penalties, but no ordering costs or sales revenue. The goal is to determine the optimal order quantity at the beginning of the horizon that minimizes the total cost at the end of the horizon with horizon length $H=3$ and initial inventory level $s_0=5$. 

\begin{itemize}
  \item \textbf{State space:} inventory level $s\in\{0,1,\dots,20\}$.
  \item \textbf{Actions:} order $a\in\{0,1,\dots,20-s\}$ units.
  \item \textbf{Demand:} each period $D\sim\mathrm{Unif}\{0,1,\dots,9\}$ i.i.d.
  \item \textbf{Period transition \& reward:}
        \begin{enumerate}
          \item After ordering, inventory level is $s+a$.
          \item Demand $D$ is realized.  
                \begin{itemize}
                  \item If $s+a\ge D$, reward (holding cost) is $-(s+a-D)$ and next state is $s+a-D$.
                  \item Otherwise: reward (lost-sale penalty) is $-10(D-(s+a))$ and next state is $0$.
                \end{itemize}
        \end{enumerate}
  \item \textbf{Optimal solution:} $a=4$ with $Q=-13.5$.
\end{itemize}
\textbf{Observations:} Results are plotted in Figs~\ref{fig:IC_prob} and \ref{fig:IC_val}. Using consistent estimators leads to performance improvement for both OCBA and UCB1-Tuned. For example, OCBA-MCTS-C reaches an empirical PCS of 0.80 at 5{,}000 rollouts, whereas OCBA–MCTS needs 18{,}000 rollouts, a speed-up of more than 3 times. OCBA-MCTS-C has the best performance. The variance estimation errors for the algorithms with the consistent procedure are lower, regardless of the selection algorithm, UCB1-Tuned or OCBA.

\begin{figure}[!t]
  \centering\includegraphics[width=.85\linewidth]{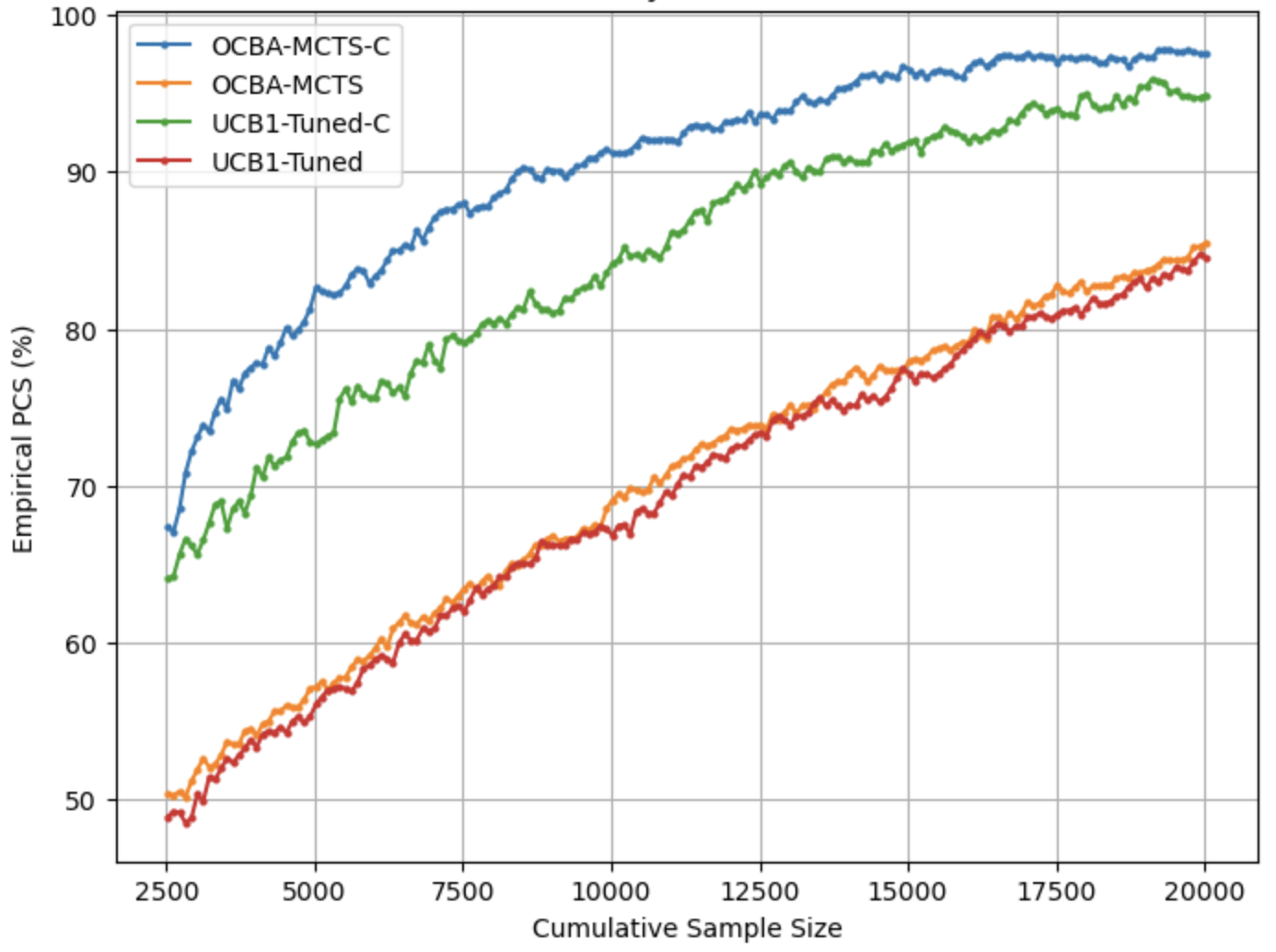}
  \caption{Inventory-Control-1 — Empirical PCS.}
  \label{fig:IC_prob}
\end{figure}

\begin{figure}[!t]
  \centering\includegraphics[width=.85\linewidth]{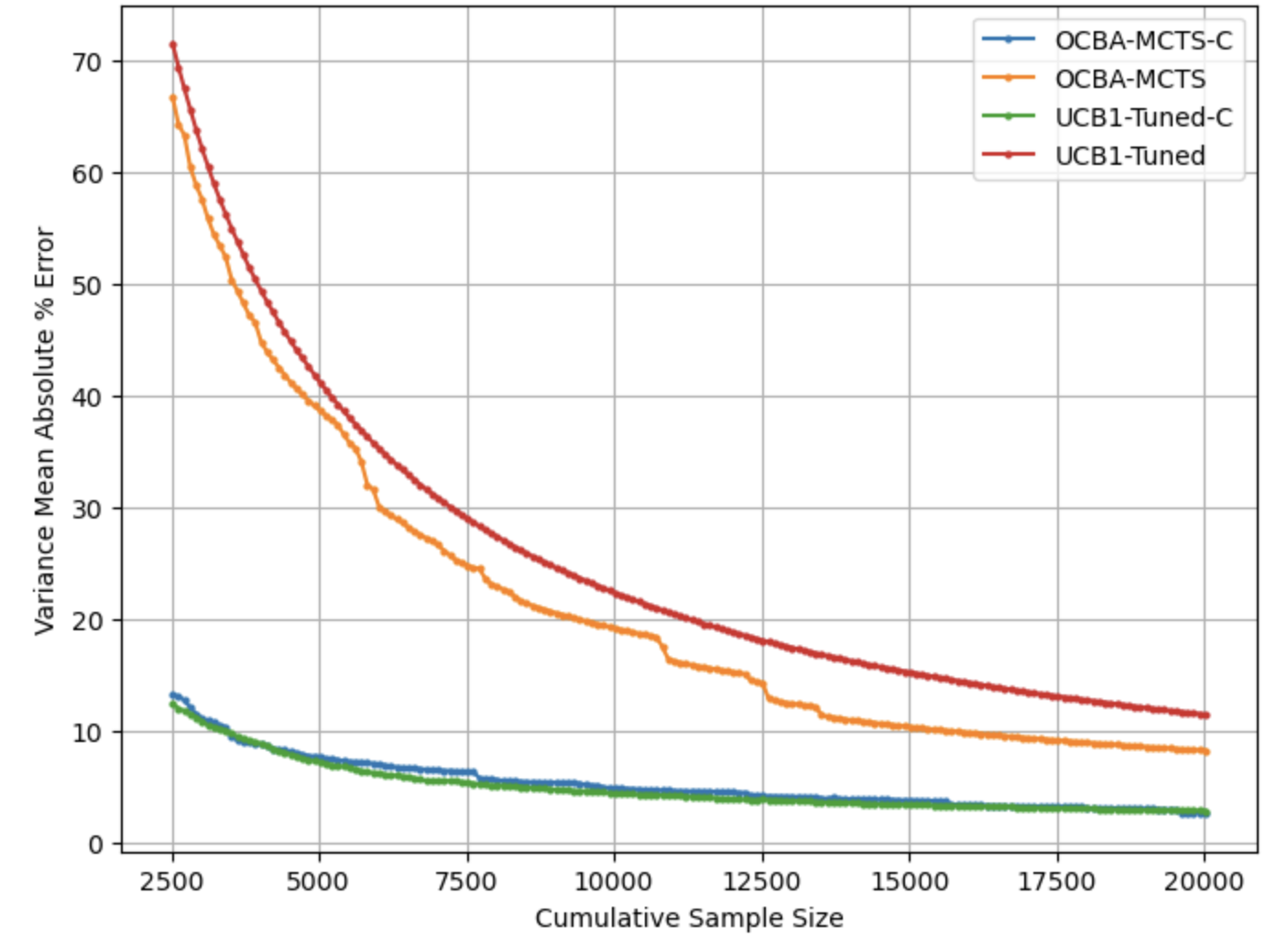}
  \caption{Inventory-Control-1 — Variance estimation error.}
  \label{fig:IC_val}
\end{figure}

\subsubsection{Inventory-2}
The setting is similar to the previous case, with the addition of ordering costs and sales revenues. The initial state is $s_0=3$ and horizon length is $H=5$.

\begin{itemize}
  \item \textbf{State space:} $s\in\{0,\dots,30\}$.
  \item \textbf{Actions:} order $a\in\{0,\dots,10\}$ subject to $s+a\le30$.
  \item \textbf{Demand:} $D\in\{3,6,9\}$ with equal probabilities.
  \item \textbf{Period transition \& reward:}
        \begin{enumerate}
          \item Order cost $c(a)=1\{a>0\}+8a$; 
          \item Demand $D$ is realized and the number of units sold is $\min(s+a,D)$.
          \item One-period reward is $10\cdot \min(s+a,D)-c(a)-0.2\cdot (s+a-D)^+$; next state is $(s+a-D)^+$.
        \end{enumerate}
  \item \textbf{Optimal solution:} $a=6$ with $Q=70.02$.
\end{itemize}

\textbf{Observations:}
Results are plotted in Figs~\ref{fig:IC2_prob} and \ref{fig:IC2_val}. This instance is harder, so overall PCS levels are lower, but using the consistent estimates again leads to significant performance gain for both OCBA and UCB1-Tuned. OCBA performs much better in terms of PCS than UCB1-Tuned, both with the consistent procedure and with the baseline. The variance estimation errors of both consistent procedures are very low, while the baseline procedures have higher estimation errors, especially the OCBA version.

\begin{figure}[!t]
  \centering\includegraphics[width=.85\linewidth]{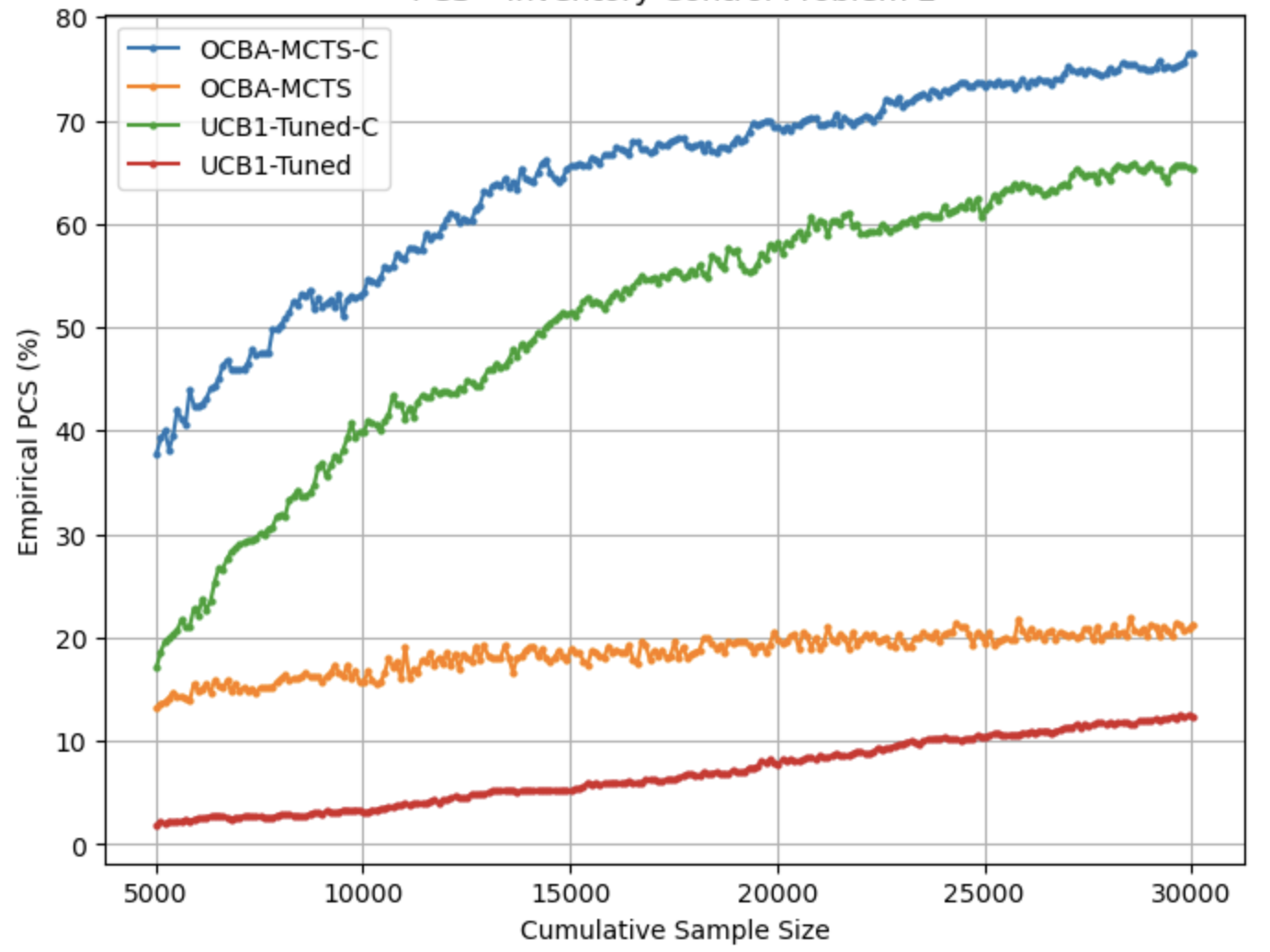}
  \caption{Inventory-Control-2 — Empirical PCS.}
  \label{fig:IC2_prob}
\end{figure}

\begin{figure}[!t]
  \centering\includegraphics[width=.85\linewidth]{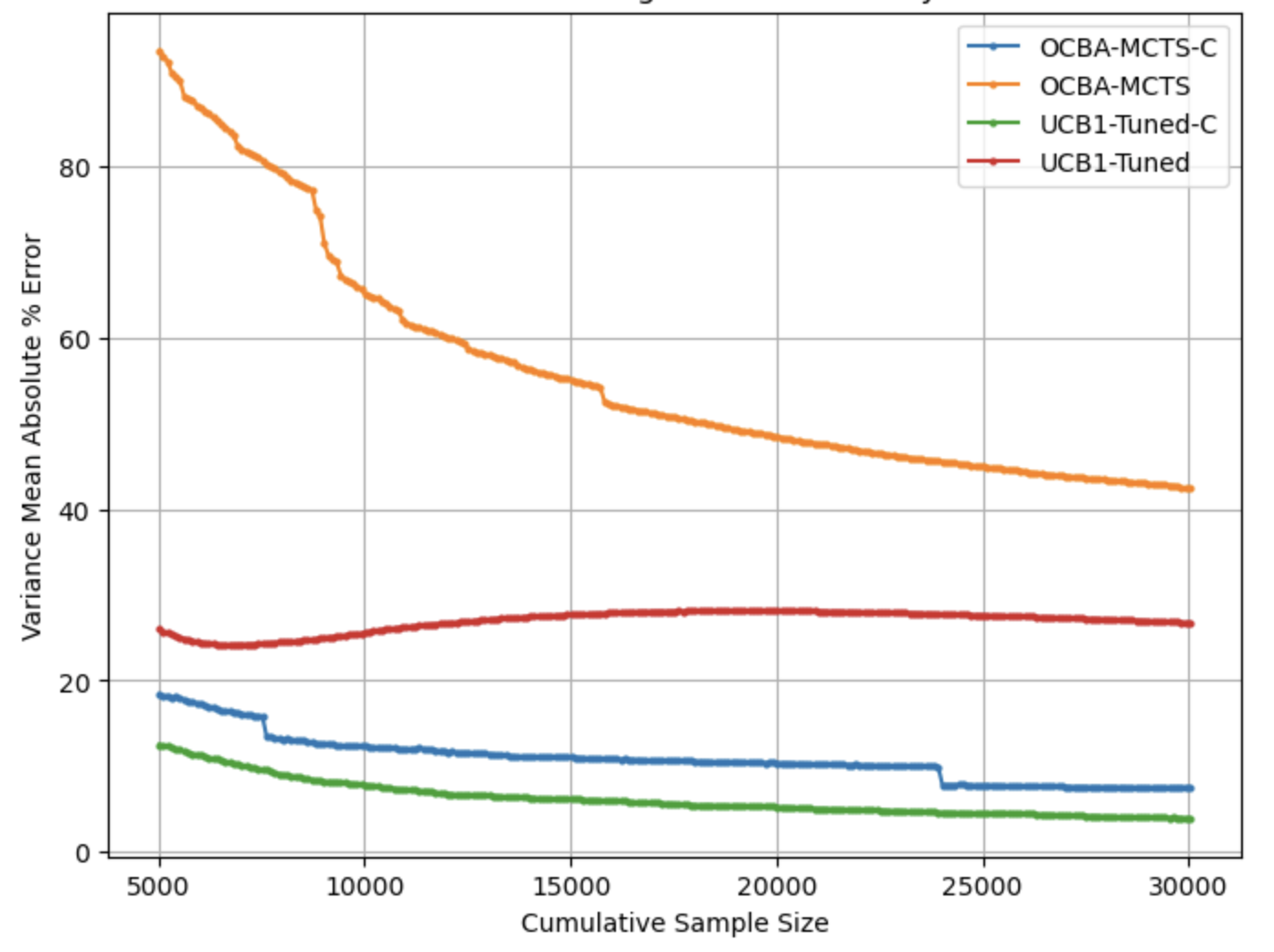}
  \caption{Inventory-Control-2 — Variance estimation error.}
  \label{fig:IC2_val}
\end{figure}

\subsubsection{Inventory-3}
Identical to Inventory Control 2 except for a longer horizon $H=10$. The optimal action is still 6, with a $Q$-value of 122.02.

\textbf{Observations:}
Results are plotted in Figs~\ref{fig:IC3_prob} and \ref{fig:IC3_val}. With a much longer horizon, this problem requires far more computational effort than the previous two cases, with only OCBA-MCTS-C achieving a PCS level exceeding 0.5 within the given budget. Using the consistent procedure again leads to significant performance improvement for both OCBA and UCB1-Tuned. The variance estimation error for the baseline OCBA-MCTS is high, while that for OCBA-MCTS-C is the lowest. Both UCB1-based algorithms maintain relatively low error compared to OCBA-MCTS; however, UCB1-Tuned-C starts with higher error than its baseline counterpart, but improves rapidly to overtake it and converge toward the accuracy of OCBA-MCTS-C after 15,000 samples.

\begin{figure}[!t]
  \centering\includegraphics[width=.85\linewidth]{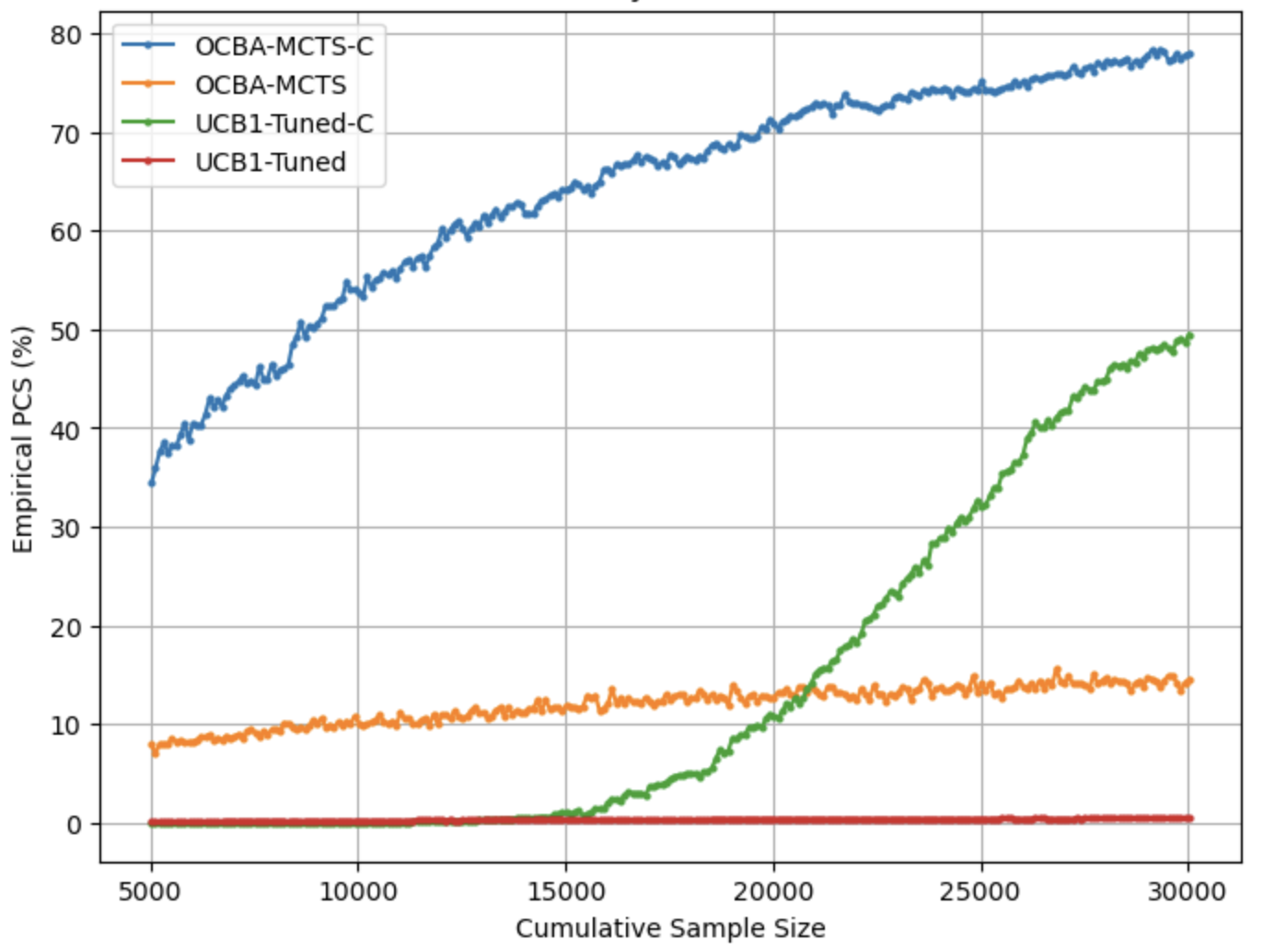}
  \caption{Inventory-Control-3 — Empirical PCS.}
  \label{fig:IC3_prob}
\end{figure}

\begin{figure}[!t]
  \centering\includegraphics[width=.85\linewidth]{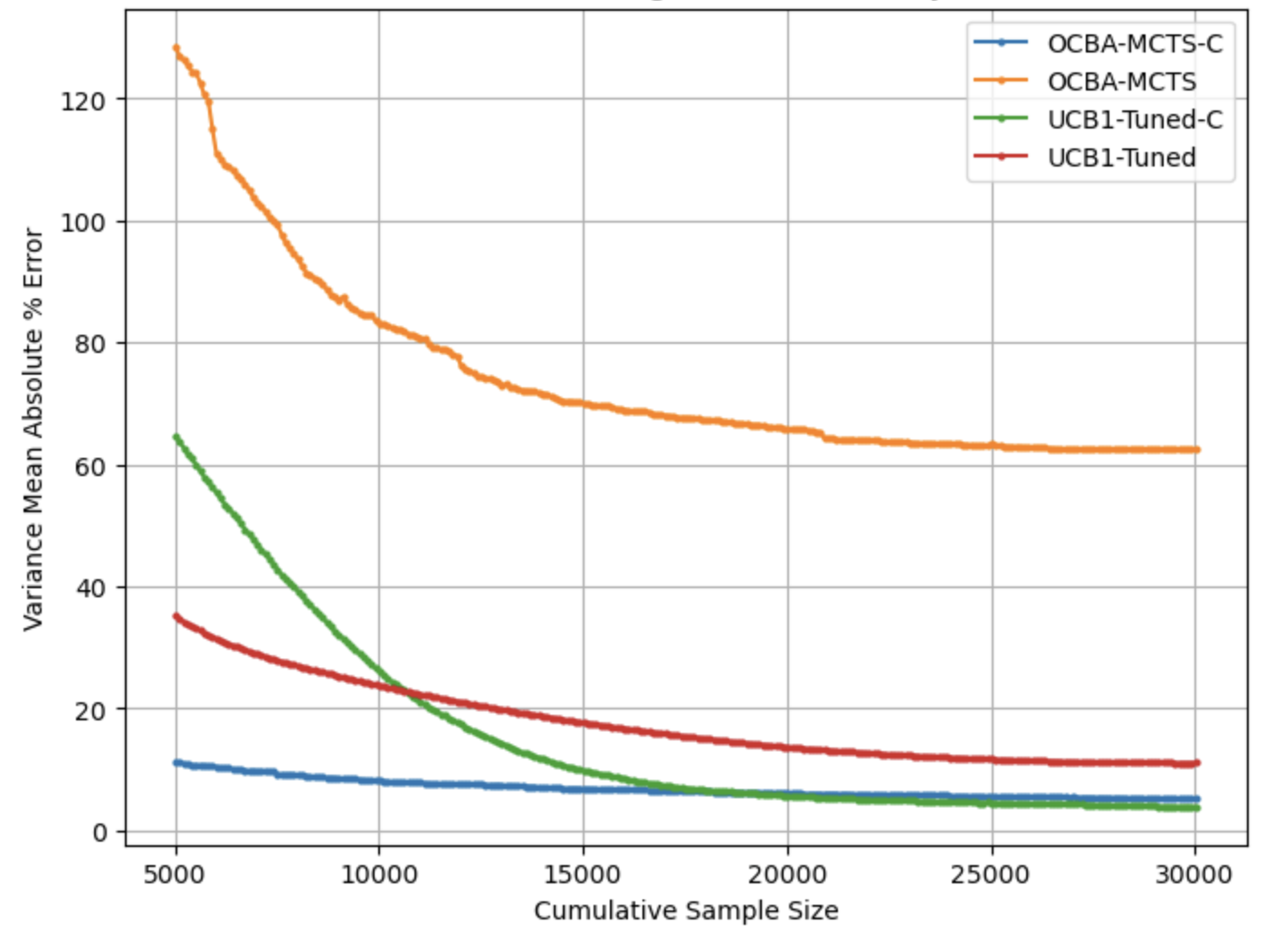}
  \caption{Inventory-Control-3 — Variance estimation error.}
  \label{fig:IC3_val}
\end{figure}

\subsection{Kidney-Paired Donation}

KPD is a transplant strategy that helps patients with willing but incompatible donors receive a kidney by exchanging donors with other incompatible pairs \cite{yue2024kidney}. In a typical KPD program, patients and their incompatible donors are matched with other pairs to form compatible exchanges, improving transplant opportunities.

We consider a finite-horizon MDP with horizon $H= 10$ for a simplified KPD setting with a two-class waiting list: urgent and stable patients. In each period, a single new patient arrives and is eligible to be matched with patients currently on the waiting list. The action in each period specifies which types of matches are permitted (e.g., urgent–urgent, urgent–stable, etc.). Each potential patient pair is independently compatible with probability (w.p) $p_{\text{match}}=0.3$. The objective is to determine the optimal matching policy, i.e., which match types to permit in each period, to maximize system performance given the current state of the waiting list.

\begin{itemize}
  \item \textbf{State:} $(n,u)$ = (\# stable patients in the waiting list, \# urgent patients in the waiting list); $s_0=(1,0)$.
  \item \textbf{Actions:} $a\in\{0,1,2,3\}$:
        \begin{itemize}
          \item 0: no match,
          \item 1: urgent–urgent only,
          \item 2: urgent–urgent and urgent–stable,
          \item 3: all match types.
        \end{itemize}
  \item \textbf{Arrivals:} one patient per period; the arrival is an urgent patient w.p. $0.3$ and stable w.p. $0.7$.
  \item \textbf{Matching \& transition:}  At the beginning of each period, a new patient arrives. Based on the assumption that every patient in the waiting list independently w.p. 0.3 is compatible with the newly arrived patient, and urgent patients are prioritized, it is easy to see that the newly arrived patient will be matched with an urgent patient w.p. $1-0.7^u$, a stable patient w.p. $0.7^u(1-0.7^n)$, or not matched w.p. $0.7^{u+n}$. If there is a match and it is allowed by the policy (depending on the choice of action), the paired transplantation occurs and the  waiting list is updated accordingly. If there is no match or the matched type is not allowed by the policy, the newly arrived patient joins the waiting list.

  \item \textbf{Reward:} $R(n',u')=-(n'+10u')$, where $(n',u') $ denotes the state after transition.
  \item \textbf{Optimal solution:} $a=2$ with $Q=-58.607$.
\end{itemize}

\begin{figure}[!t]
  \centering\includegraphics[width=.85\linewidth]{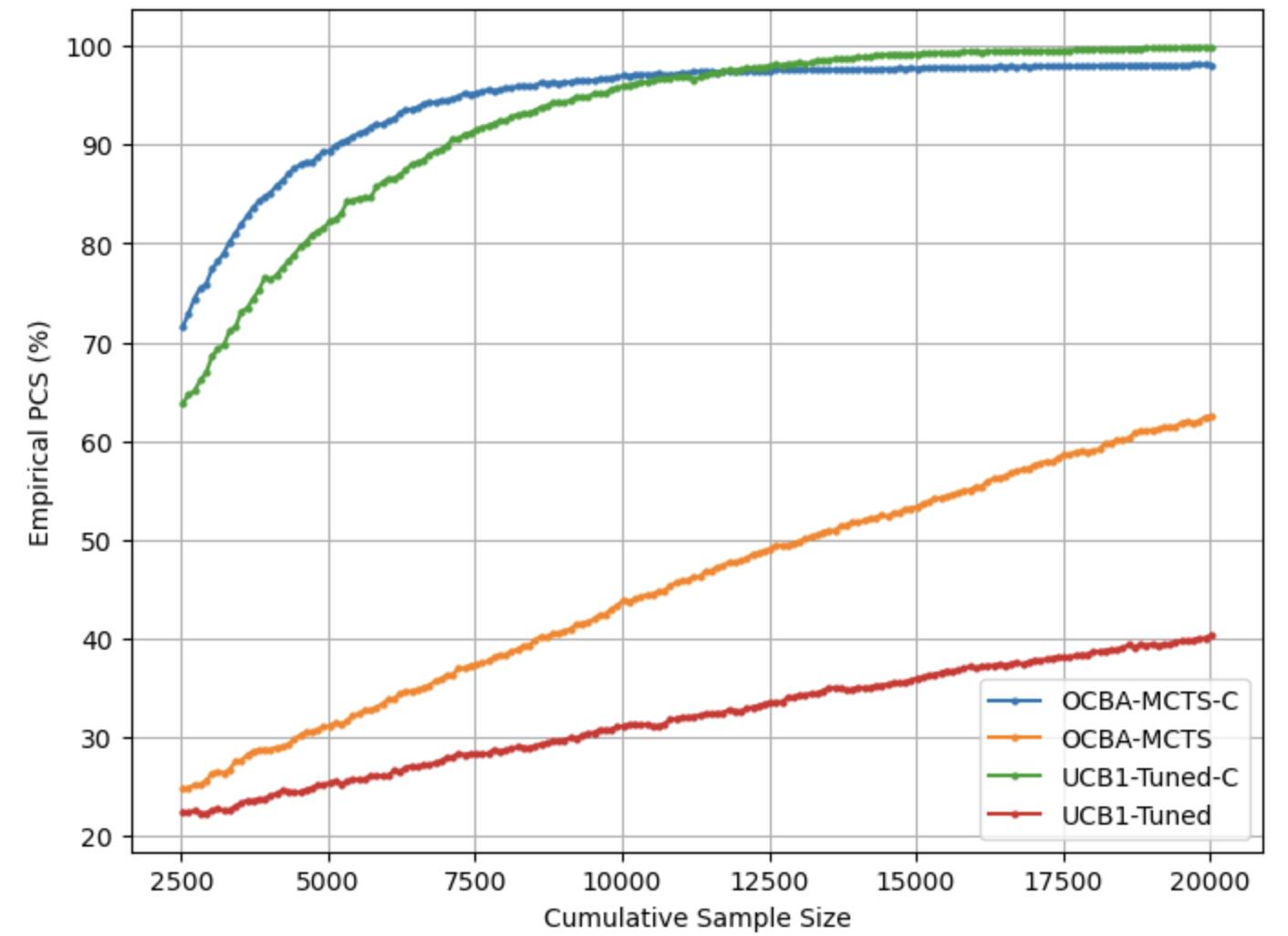}
  \caption{KPD — Empirical PCS.}
  \label{fig:KPD_prob}
\end{figure}

\begin{figure}[!t]
  \centering\includegraphics[width=.85\linewidth]{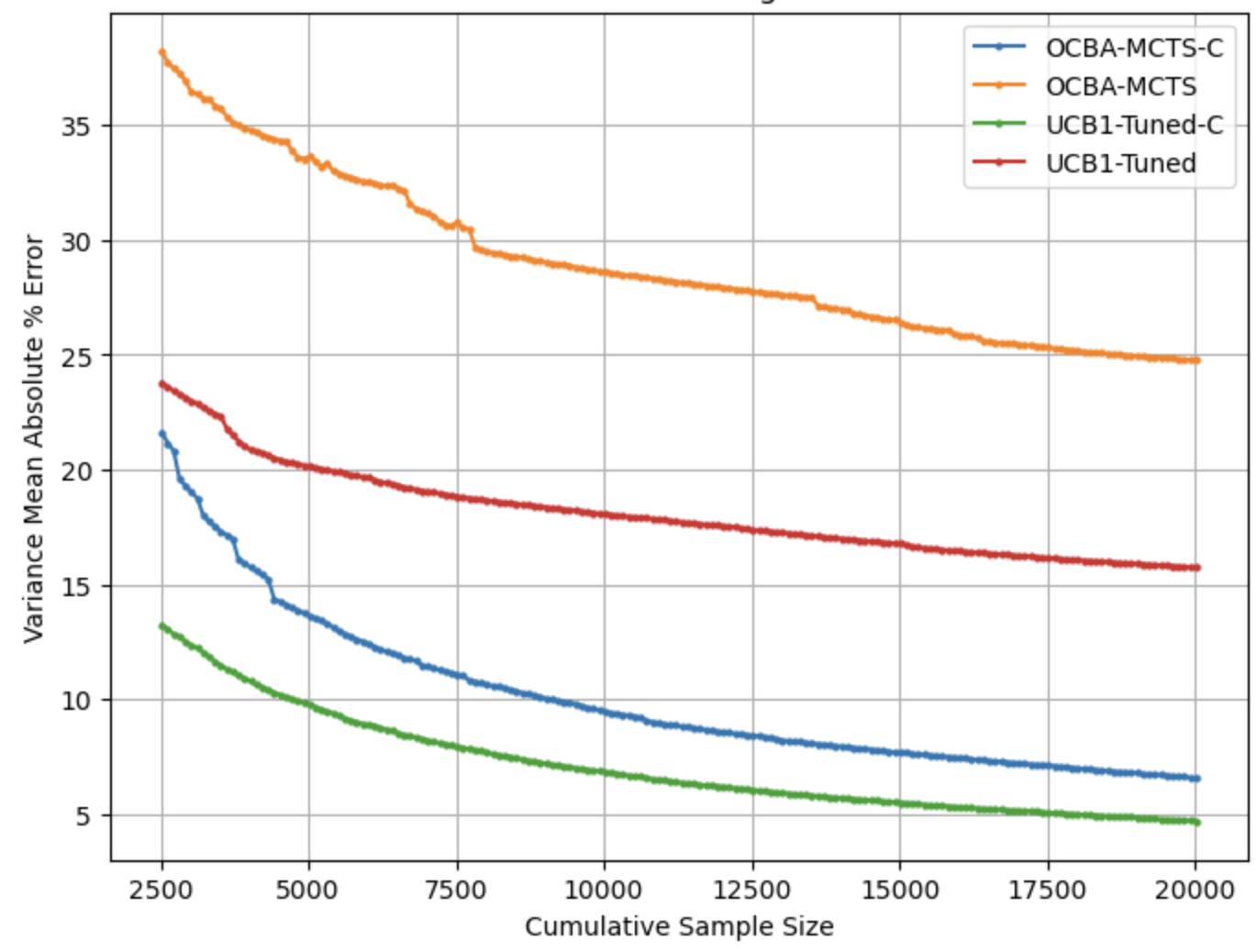}
  \caption{KPD — Variance estimation error.}
  \label{fig:KPD_val}
\end{figure}



\textbf{Observations:}
Results are plotted in Figs~\ref{fig:KPD_prob} and \ref{fig:KPD_val}. Similar to the results from the inventory control examples, using the consistent procedure leads to significant performance improvement for both OCBA and UCB1-Tuned. For this problem, OCBA–MCTS–C and UCB1-Tuned-C have similar performance, with OCBA-MCTS-C leading earlier but UCB1-Tuned catching up after 10,000 rollouts and then exceeding OCBA-MCTS-C by a small margin. For variance estimation, UCB1-Tuned-C outperforms OCBA-MCTS-C, but the consistent procedures are always better than their baseline counterparts.

\vspace{-4pt}
\subsection{Summary}
The consistent variance estimators deliver more accurate variance estimates in all test problems. The accurate variance estimates significantly improve both OCBA and UCB1-Tuned's performance across all test problems. OCBA is also shown to outperform UCB1-Tuned on PCS in almost all cases, especially when the problem is more difficult to solve. This is consistent with the findings in \cite{li2022optimal,fu2024tutorial} that a selection rule maximizing PCS is more effective in identifying the best action for a finite-horizon MDP.


\vspace{-4pt}
\section{Conclusion}\label{sec:conclusion}
In this paper, we investigated variance estimation for state--action value ($Q$-function) estimators in tree-structured algorithms for solving finite-horizon Markov decision processes (MDPs). Our main contributions include (i) a variance decomposition for $Q$-function estimators into three components arising from uncertainty in rewards, transition dynamics, and next-period value estimates; (ii) a consistent recursive variance estimator derived based on this decomposition; and (iii) an implementation of this estimator using node-local statistics that can be iteratively updated without storing full trajectories or rebuilding the tree.

We incorporated the new variance estimator into two MCTS selection rules based on OCBA and UCB1-Tuned. Numerical experiments on an inventory control problem and a kidney paired donation problem show that the consistent estimator substantially improves variance estimation accuracy and sample efficiency for both selection rules, resulting in more reliable and effective MCTS algorithms.

Because the proposed estimator depends only on local value and reward statistics at each node, it can be integrated with many existing tree-structured planning and reinforcement learning algorithms without changing their overall structure. Our analysis in this paper focused on finite-horizon, finite-state problems, and extending the approach rigorously to settings with large or continuous state and action spaces---for example, via function approximation---is an important direction for future work. Another promising direction would be using the consistent variance estimates to guide dynamic branching or early stopping decisions, explicitly optimizing the allocation of simulation budget across the planning horizon.

\vspace{-8pt}
\section*{References}

\vspace{-18pt}
\bibliographystyle{IEEEtran}   
\bibliography{references}      

\vspace{-22pt}
\begin{IEEEbiography}[{\includegraphics[width=1in,height=1.25in,clip,keepaspectratio]{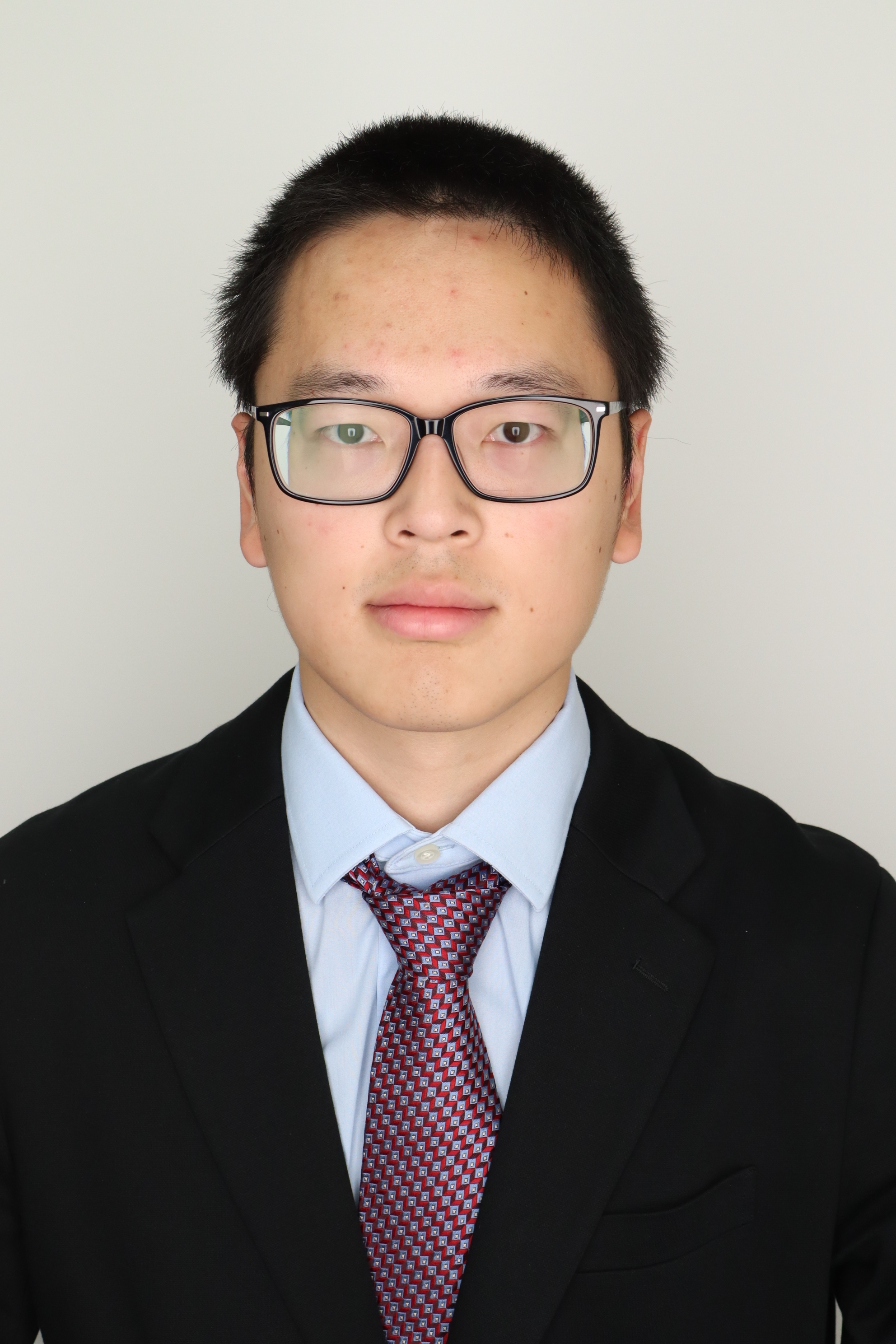}}]{Zhenyu Yue} received the bachelor's degree in mathematics and computer science from University of Maryland, College Park, MD, USA, in 2018. He is a Ph.D. candidate of applied math at University of Maryland, College Park, MD, USA. His research interests include stochastic processes and optimization.
\end{IEEEbiography}

\begin{IEEEbiography}{Jie Xu}(Senior Member, IEEE) received the B.S. degree in electrical engineering from Nanjing University, Nanjing, China, in 1999, the M.E. degree in electrical engineering from Shanghai Jiaotong University, Shanghai, China, in 2002, the M.S. degree in computer science from The State University of New York, Buffalo, NY, USA, in 2004, and the Ph.D. degree in industrial engineering and management sciences from Northwestern University, Evanston, IL, USA, in 2009. He is currently an Associate Professor of Systems Engineering and Operations Research with George Mason University, Fairfax, VA, USA. His research interests are data analytics, stochastic simulation and optimization, with applications in cloud computing, manufacturing, and power systems.
\end{IEEEbiography}

\begin{IEEEbiography}{Chun-Hung Chen} (Fellow, IEEE) received the Ph.D. degree in decision and control from Harvard University, Cambridge, MA, USA, in 1994. He is a Professor of Systems Engineering and Operations Research with George Mason University (GMU), Fairfax, VA, USA. Before joining GMU, he was an Assistant Professor with the University of Pennsylvania, Philadelphia, PA, USA. He was also a Professor of Electrical Engineering and Industrial Engineering with National Taiwan University, Taipei, Taiwan, from 2011 to 2014. He is the author of two books, including a best seller: Stochastic Simulation Optimization: An Optimal Computing Budget Allocation (World Scientific, 2011). Prof. Chen has served as a Department Editor for IIE Transactions and Asia–Pacific Journal of Operational Research, an Associate Editor for IEEE Transactions on Automation Science and Engineering, IEEE Transactions on Automatic Control, and Journal of Systems Science and Systems Engineering, an Area Editor for Journal of Simulation Modeling Practice and Theory, and an Advisory Editor for International Journal of Simulation and Process Modeling and Journal of Traffic and Transportation Engineering. 
Dr. Chen is a Fellow of the Institute for Operations Research and the Management Sciences (INFORMS).
\end{IEEEbiography}

\begin{IEEEbiography}{Hadi El-Amine} received the B.E. degree in Electrical and Computer Engineering (2009) and the M.E. degree in Engineering Management (2012) from the American University of Beirut, Beirut, Lebanon, and the Ph.D. degree in Industrial and Systems Engineering from Virginia Tech, Blacksburg, VA, USA (2016). He is an associate professor of Systems Engineering and Operations Research at George Mason University, Fairfax, VA, where his research focuses on stochastic optimization and its applications to healthcare and public policy. He is a member of the Institute of Operations Research and Management Science (INFORMS) and the Institute of Industrial and Systems Engineers (IISE).
\end{IEEEbiography}

\begin{IEEEbiography}{Michael C. Fu} (Fellow, IEEE) received degrees in mathematics and electrical engineering and computer science from MIT, Cambridge, MA, USA, in 1985, and the Ph.D. degree in applied math from Harvard University, Cambridge, MA, USA, in 1989. Since 1989, he has been with the University of Maryland, currently holding the Smith Chair of Management Science. He also served as Program Director of the Operations Research Program at the U.S.~National Science Foundation. His research interests include stochastic modeling, simulation optimization, and stochastic gradient estimation. Dr. Fu is a Fellow of the Institute for Operations Research and the Management Sciences (INFORMS).
\end{IEEEbiography}

\end{document}